\documentclass[11pt,letterpaper]{article}
\usepackage[T1]{fontenc}
\usepackage{lmodern}
\usepackage[margin=1in]{geometry}
\usepackage{amsmath,amssymb,amsthm,mathtools}
\usepackage{microtype}
\usepackage{enumitem}
\usepackage{xcolor,float,tikz}
\usepackage[colorlinks=true,linkcolor=blue!45!black,citecolor=blue!45!black,urlcolor=blue!45!black]{hyperref}
\newtheorem{theorem}{Theorem}[section]
\newtheorem{lemma}[theorem]{Lemma}
\newtheorem{proposition}[theorem]{Proposition}
\newtheorem{corollary}[theorem]{Corollary}

\newtheorem{conjecture}[theorem]{Conjecture}
\theoremstyle{definition}

\newtheorem{remark}[theorem]{Remark}
\newcommand{\F}{\mathbb F}
\newcommand{\E}{\mathbb E}
\newcommand{\one}{\mathbf 1}
\DeclareMathOperator{\RM}{RM}
\DeclareMathOperator{\Ber}{Ber}
\DeclareMathOperator{\rank}{rank}
\DeclareMathOperator{\diag}{diag}
\DeclareMathOperator{\supp}{supp}

\DeclareMathOperator{\poly}{poly}

\newcommand{\KeyGen}{\mathsf{KeyGen}}
\newcommand{\Enc}{\mathsf{Enc}}
\newcommand{\Dec}{\mathsf{Dec}}
\setlist{itemsep=3pt,topsep=5pt,leftmargin=*}
\title{Expanding Codes with Polynomial Stretch\footnote{This work subsumes a previous work by the same authors titled ``Public Key Encryption from High-Corruption Constraint Satisfaction Problems.''}}
\author{Isaac M Hair\thanks{\texttt{isaacmhair@gmail.com}}\\UCSB, UCLA
\and Amit Sahai\thanks{\texttt{sahai@cs.ucla.edu}}\\UCLA}
\date{}
\begin{document}
\hypersetup{pageanchor=false}
\maketitle
\thispagestyle{empty}
\vspace{-2em}
\begin{abstract}
We give a polynomial-stretch, constant-arity counterpart to the expanding code construction of B{\l}asiok, Lou, Rosen, and Sudan (FOCS 2026). For every fixed $\varepsilon\in(0,1/2)$, we construct binary matrices $\mathbf{G}$ with $n$ columns and $m = \text{poly}(n)$ rows, each of weight $k=O(\varepsilon^{-6})$, such that the union of the supports of any subset $S$ of at most $n^{1-\varepsilon}$ rows has size at least~$(1~-~\varepsilon)k\vert S\vert$. Nevertheless, a polynomial-time test accepts every codeword $\mathbf{G}\mathbf{x}$ corrupted arbitrarily in at most a $1/2-\varepsilon$ fraction of its coordinates, while rejecting a uniformly random vector with probability $1-\exp(-\Omega_\varepsilon(m))$. The construction is unconditional and randomized. 

We give two applications. First, we obtain a family of (nearly lossless) constant locality expanders on which Goldreich's PRG with the XOR-MAJ predicate can be distinguished from uniform in polynomial time, well within the small-bias parameter regime, as well as the hard regime for known algebraic-immunity/annihilator attacks and algebraic-degree linearization attacks, as described by Applebaum and Lovett (STOC 2016, SICOMP 2018).

Second, we obtain public key encryption from either the constant-noise $k$XOR conjecture over random graphs or the Goldreich conjecture with XOR-MAJ over random graphs, together with new planted-hypergraph conjectures. We prove low-degree lower bounds for the planting conjectures, providing evidence for  subexponential hardness, and give lower bounds against a variety of structural attacks.

\emph{The PKE schemes were discovered by the authors. The code construction was discovered by Codex / ChatGPT 6 Astra Ultra using harness components from the UCLA Moonshot Harness Project. The lower bounds were proved both by the authors and by Codex / ChatGPT 6.2 Astra Max. See the end of the introduction for a discussion.}
\end{abstract}

\newpage
\setcounter{page}{1}
\hypersetup{pageanchor=true}

\section{Introduction}\label{sec:intro}
Noisy $k$XOR is the problem of distinguishing sparse linear equations generated from a hidden assignment, and then corrupted by noise, from equations with independent uniform right-hand sides. The problem has been studied both as a source of cryptographic hardness and as an example of a statistical--computational gap: a noisy planted system can be statistically far from uniform even in regimes where no efficient test is known~\cite{alekhnovich2003more,applebaum10}. Expansion is a natural candidate for explaining this difficulty: every small collection of equations should involve almost as many distinct variables as the total number of variable occurrences in those equations. This condition rules out small linear dependencies and underlies several lower bounds for restricted algorithms. Nevertheless, expansion alone need not rule out efficient algorithms that exploit other structure in the equations.

\paragraph{Recognizing Noisy Codewords.}
We study the stronger requirement of recognizing arbitrary codewords corrupted by arbitrary errors of bounded weight. Let $\mathbf{G}\in\F_2^{m\times n}$ have exactly $k$ nonzero entries in each row, and fix an error fraction $\eta\in(0,1/2)$. We seek a test that accepts
\[
 \mathbf{G}\mathbf{x}+\mathbf{e}\qquad\text{for every }\mathbf{x}\in\F_2^n
 \text{ and }\mathbf{e}\in\F_2^m\text{ with }|\operatorname{supp}(\mathbf{e})|\le\eta m,
\]
while rejecting a uniformly random vector with high probability. Here $\supp(\mathbf{e})$ denotes the set of nonzero coordinates of $\mathbf{e}$. Equivalently, the columns of $\mathbf{G}$ generate a binary linear code, and we want to recognize every vector within relative Hamming distance $\eta$ of this code. We use boldface for matrices and vectors. Vectors are columns, and we list their coordinates from top to bottom. A generator matrix has one row per codeword coordinate, and its columns need not be linearly independent. The \emph{stretch} is $m/n$; polynomial stretch means that the number $m$ of equations is polynomial in $n$.

For a subset $S$ of rows, write $\Gamma_{\mathbf{G}}(S)$ for the columns containing a nonzero entry in one of those rows. We say that $\mathbf{G}$ is \emph{$(\gamma,k,t)$-expanding} if each row has weight $k$ and
\[
 |\Gamma_{\mathbf{G}}(S)|\ge\gamma k|S|\qquad\text{for every }|S|\le t.
\]
The quantity $1-\gamma$ is the expansion loss. When it is small, almost all of the $k|S|$ nonzero entries in the rows indexed by $S$ lie in distinct columns. In particular, when $\gamma>1/2$, every set of at most $t$ rows is linearly independent: a nonempty dependence would require every neighboring column to occur at least twice. Expansion is also central to lower bounds for the sum-of-squares hierarchy, which searches for bounded-degree polynomial certificates that a constraint system has no solution~\cite{grigoriev2001linear,schoenebeck2008linear,kothari2017sum}. The conjecture that expansion alone suffices for hardness of noisy linear equation solving goes back to Alekhnovich~\cite{alekhnovich2003more}; related questions also arise for local pseudorandom generators, whose output bits each depend on only a few input bits~\cite{oliveira2018expander}.

B{\l}asiok, Lou, Rosen, and Sudan~\cite{BLRS26} recently gave counterexamples to this principle for noisy $k$XOR. Their unconditional construction has a nearly linear expansion cutoff (the value~$t$ above) and can tolerate constant noise, but the stretch of their matrix is quasipolynomial, and the row arity is polylogarithmic. They obtain polynomial stretch (with polylogarithmic row arity) under a finite-length erasure-recovery conjecture for Reed--Muller codes, but in this setting they are only able to distinguish codewords with \emph{inverse-polynomial rate noise} from random vectors. We unconditionally construct strongly expanding matrices with polynomial stretch and constant arity, where codewords with arbitrary errors of relative weight at most $1/2-\varepsilon$ can be distinguished from random vectors.

\subsection{Our Results}
Our main result is the following.

\begin{theorem}\label{thm:main}
For every constant $\varepsilon\in(0,1/2)$, there is an integer $k=O(\varepsilon^{-6})$ and a randomized $\poly_\varepsilon(n)$-time algorithm that, for every sufficiently large $n$, outputs $\mathbf{G}\in\F_2^{m\times n}$, where $m=n^{O_\varepsilon(1)}>n$ depends only on $n$ and $\varepsilon$ and $m < n^{\varepsilon k}$, such that:
\begin{enumerate}
\item Every row of $\mathbf{G}$ has weight exactly $k$, and with probability at least $1-n^{-1/\varepsilon}$ the matrix $\mathbf{G}$ is $(1-\varepsilon,k,n^{1-\varepsilon})$-expanding.
\item There is a deterministic $\poly_\varepsilon(m,n)$-time test $D$ such that, for every output $\mathbf{G}$,
\[
 D(\mathbf{G}\mathbf{x}+\mathbf{e})=1\quad\text{for every }\mathbf{x}\in\F_2^n,\ \mathbf{e}\in\F_2^m
 \text{ with }|\operatorname{supp}(\mathbf{e})|\le(1/2-\varepsilon)m,
\]
and $\Pr[D(\mathbf{u}_m)=1]\le\exp(-\Omega_\varepsilon(m))$, where $\mathbf{u}_m$ denotes a vector sampled from the uniform distribution $U_m$ on $\F_2^m$.
\end{enumerate}
\end{theorem}

Thus the number of equations is polynomial, while the arity and tolerated error fraction are constants. The expansion loss can be any fixed positive constant, and expansion holds up to $n^{1-\varepsilon}$ rows. Our construction is randomized; in contrast, the unconditional construction of~\cite{BLRS26} is explicit and has expansion loss $o(1)$. Section~\ref{sec:warmup} starts with a simple warm-up construction based on random lines in the plane. Its row weights are logarithmic with high probability, the stretch is polynomial, and its test already tolerates a constant fraction of arbitrary errors. Section~\ref{sec:codes} gives the constant-arity construction and proves Theorem~\ref{thm:main}. In both our warm-up and our constant-arity construction, the test computes the rank of a matrix whose entries are linear forms in the received word. This matrix vanishes on codewords, each bit flip changes its rank by at most one, and a uniform word typically gives large rank.

\paragraph{Expansion and Local Pseudorandomness.}
The connection between expansion and hardness extends beyond noisy linear equations. In proof complexity, expansion yields degree lower bounds for polynomial calculus, a proof system that derives contradictions by manipulating polynomial equations~\cite{ben1999random}; degree lower bounds can in turn imply lower bounds on proof size~\cite{impagliazzo1999lower}. For semidefinite hierarchies, the parity lower bounds of~\cite{grigoriev2001linear,schoenebeck2008linear} were followed by extensions to more general constraint satisfaction problems and sharper density-dependent bounds~\cite{tulsiani2009csp,mori2016lower,kothari2017sum}. These results are part of the broader study of average-case constraint satisfaction, including Feige's refutation hypothesis and hypotheses about the optimality of semidefinite relaxations~\cite{feige2002relations,barak2013optimality}.

For local cryptographic functions, Goldreich proposed applying a fixed predicate to overlapping subsets of input bits selected by an expander~\cite{goldreich2011candidate}. Subsequent work established small bias, reductions from random local one-way functions, and lower bounds for semidefinite hierarchies in related settings~\cite{mossel2006biased,applebaum2012pseudorandom,odonnell2014goldreich}. A simple modification of the construction in Theorem~\ref{thm:main} gives a counterexample to pseudorandomness in this setting for specific expanders, using the standard XOR-MAJ predicate with equally many inputs to its parity and majority gates.

For an odd positive integer $\ell$, define the XOR-MAJ predicate $P_\ell:\F_2^{2\ell}\to\F_2$ by
\[
 P_\ell(u_1,\ldots,u_{2\ell})
 =\bigoplus_{j=1}^{\ell}u_j
 \;\oplus\;
 \mathbf1\!\left[\sum_{j=\ell+1}^{2\ell}u_j>\ell/2\right].
\]
A Goldreich generator on $N$ inputs is specified by a list $H$ of $m$ ordered tuples, each containing $2\ell$ distinct input indices. Its $i$th output applies $P_\ell$ to the bits indexed by the $i$th tuple; write $f_H:\F_2^N\to\F_2^m$ for this map. The incidence matrix of $H$ records which input indices occur in each tuple, and we apply our expansion definition to this matrix. Thus $\Gamma_H(S)$ is the set of input indices appearing in the tuples indexed by $S$.

\begin{corollary}[Goldreich's Generator with XOR-MAJ]\label{cor:xor-maj}
For every constant $\varepsilon\in(0,1/2)$, there is an odd integer $\ell=O(\varepsilon^{-6})$ with the following property. For every sufficiently large $N$, a randomized $\poly_\varepsilon(N)$-time algorithm samples a list $H$ as above with $N<m<N^{\varepsilon\ell}$ tuples. With probability at least $1-N^{-1/\varepsilon}$, its incidence matrix is $(1-\varepsilon,2\ell,N^{1-\varepsilon})$-expanding. For every output $H$, a deterministic $\poly_\varepsilon(m,N)$-time test $D$ satisfies
\[
 \Pr[D(f_H(\mathbf{u}_N))=0]\le \exp(-\Omega_\varepsilon(N^{\varepsilon/2})),
 \qquad
 \Pr[D(\mathbf{u}_m)=1]\le \exp(-\Omega_\varepsilon(m)),
\]
where $\mathbf{u}_N$ denotes a vector sampled from the uniform distribution $U_N$ on $\F_2^N$, and $\mathbf{u}_m$ denotes a vector sampled from the uniform distribution $U_m$ on $\F_2^m$. The probabilities are over the uniform seed and the uniform comparison word, respectively.
\end{corollary}

The expansion loss can therefore be arbitrarily small, even when the output-length exponent is an arbitrarily small fraction of the locality. See Section~\ref{sec:xor-maj} for the construction.

\paragraph{Earlier Results on XOR-MAJ.}
XOR-MAJ has been a prominent candidate for Goldreich's construction. Applebaum and Lovett~\cite{applebaum2016algebraic} proposed it as a concrete challenge after studying the requirements imposed by linear and algebraic attacks. Applebaum and Raykov~\cite{applebaum2016fast} then singled out its balanced form in conjectures asserting security on every sufficiently expanding hypergraph, and used expander-based conjectures to construct fast pseudorandom functions. XOR-MAJ has also been a recurring target of concrete cryptanalysis~\cite{couteau2018concrete}, and remains a standard example in work on the security of random local functions~\cite{tan2025improved}. Our corollary shows that expansion alone does not suffice for pseudorandomness even for this predicate, with constant locality and polynomial stretch.

Oliveira, Santhanam, and Tell~\cite{oliveira2018expander} established a related failure of expansion-based security at quasipolynomial stretch. Their unconditional result applies to a broad class of predicates; the highlighted example is XOR-APPROX-MAJ, where approximate majority need only agree with majority when the fraction of ones is outside a fixed interval around $1/2$. Their expanders are nonexplicit, and their polynomial-stretch results require additional unproven hypotheses. Corollary~\ref{cor:xor-maj} gives the first unconditional polynomial-stretch counterexample for exact XOR-MAJ over good expanders with constant locality\footnote{That is, $(\gamma, k, t)$-expanders with $\gamma \gg 1/2$ and $k=O(1)$.}, and does so with arbitrarily small constant expansion loss and expansion up to $N^{1-\varepsilon}$ output coordinates. Both sampling the graph and distinguishing its output take polynomial time.

Applebaum and Lovett~\cite{applebaum2016algebraic} give algebraic attacks on \emph{random} Goldreich generators when the output-length exponent is sufficiently large relative to the predicate. D'Antona, M{\'e}aux, and {\"U}nal~\cite{dantona2026common} obtain polynomial-time inversion with inverse-polynomial success when the majority arity is linear in the input length. Our graphs have $\log_N m<\ell/36$ for large $N$ (Section~\ref{sec:xor-maj}). At this stretch the parity and majority arities meet the bounds in Applebaum and Lovett's challenge for random instances. These bounds account for parity tests (small bias means that every nonempty output parity is nearly fair), low-degree polynomials vanishing on one fiber of the predicate (annihilators), and linearization, which treats the monomials in the equations as separate unknowns~\cite[Section~1.3]{applebaum2016algebraic}. Our distinguisher uses the particular choice of expander.

\paragraph{Expander Codes vs. Expanding Codes.}
In the \emph{expander codes} of Sipser and Spielman~\cite{sipser1996expander}, expansion is present in the \emph{parity check matrix}, not the generator matrix. Write $\mathbf{H}$ for the matrix whose columns are the parity checks. The code consists of the words $\mathbf{w}$ satisfying $\mathbf{H}^{\mathsf T}\mathbf{w}=\mathbf{0}$. The bipartite graph of $\mathbf{H}$ joins word coordinates to parity checks, as in Tanner's construction~\cite{tanner1981recursive}. Sufficient expansion from coordinates to checks on sets of linear size gives constant relative minimum distance and a linear-time decoder for a constant fraction of arbitrary bit flips. Z\'emor improved the decoding radius for a bipartite construction~\cite{zemor2001expander}. Feldman et al.\ analyzed linear programming decoding~\cite{feldman2007lp}, and Viderman obtained linear-time decoding under weaker vertex expansion~\cite{viderman2012decoding}.

In our \emph{expanding codes}, the generator matrix $\mathbf{G}$ itself expands from output coordinates to input variables. Whenever $\mathbf{G}$ satisfies Theorem~\ref{thm:main}'s expansion guarantee, the code has inverse-polynomial relative minimum distance. Indeed, any $t=\lfloor n^{1-\varepsilon}\rfloor$ rows meet at least $(1-\varepsilon)kt$ nonzero columns. But the whole matrix contains only $km$ ones, so some nonzero column is a codeword of relative weight at most $1/((1-\varepsilon)t)$. Thus  decoding from any fixed positive fraction of arbitrary bit flips is impossible. Our test accepts every word within relative distance $1/2-\varepsilon$ of a codeword and rejects a uniform word with high probability.

\subsection{Applications to Public Key Encryption}
In this time of algorithmic surprise (e.g.~\cite{OpenAI2026ETH}), it is especially important for us to be exploring new avenues for building public-key encryption that may remain secure even if existing  proposals are broken.
Applebaum, Barak, and Wigderson~\cite{applebaum10} proposed looking to average-case combinatorial problems to broaden the foundations of public key encryption. Noisy constraint systems are a natural place to look because their average-case complexity is already studied independently of encryption. To use such a system for encryption, we need a way to give the receiver an advantage even at a noise rate at which the public instance is conjectured to be hard.

A central difficulty in code-based public key encryption is hiding the structure that makes decryption efficient. McEliece's construction~\cite{mceliece1978public} hides an efficiently decodable algebraic code, while Alekhnovich's construction~\cite{alekhnovich2003more} uses noisy linear systems. We still have relatively few well-studied ways to conceal useful code structure. The recent work of Ghoshal, Ishai, Jain, and Sun~\cite{GIJS26} reinforces the need for new approaches: they give a provable quasipolynomial-time distinguisher for binary Goppa public keys in the asymptotic Classic McEliece regime, together with heuristic quasipolynomial-time decryption and key-recovery attacks.

\paragraph{Hiding the Generator Matrix.}
Our code provides a different way to arrange the information needed for decryption. We place its rows among the rows of an independently sampled random sparse matrix, hiding their positions and original order. The combined matrix is public. To encrypt zero, generate a noisy codeword for the larger matrix; to encrypt one, generate a uniform vector. The secret key identifies the coordinates belonging to the hidden code. On those coordinates the receiver can apply the distinguishing algorithm from Theorem~\ref{thm:main}.

Planting a hidden combinatorial object has a longer cryptographic history, including the hidden-clique construction of Juels and Peinado~\cite{juels2000hiding}. More directly, Applebaum, Barak, and Wigderson hide a failure of expansion inside noisy equations or local generators~\cite{applebaum10}; Bogdanov, Kothari, and Rosen revisit this approach using low-degree lower bounds~\cite{bogdanov2023public}; and Ghosal, Hair, Jain, and Sahai combine planted clique with noisy linear equations over expanders~\cite{ghosal2025using}. Our hidden object instead expands nearly optimally on small edge sets. Its noisy codewords become recognizable when their coordinates are known.

Two conjectures suffice to prove security. First, the public matrix must hide the structured code. View a weight-$k$ row as a $k$-element hyperedge on its column indices. For each sufficiently small fixed $\varepsilon>0$ and sufficiently large fixed $k$, let $H$ be any $(1-\varepsilon,k,n^{1-\varepsilon})$-expanding hypergraph on $n$ vertices with $O(n^{\varepsilon k})$ edges, where the implied constant is fixed. Inject its vertices uniformly into $N=\lceil n^{C_1}\rceil\gg n^2$ vertices, and add the image of its edges to a large collection of random background edges on the $N$ vertices. The background is sampled so that it contains $\Theta_k(N^{C_2\varepsilon k})$ edges in expectation, where $C_2$ can be taken to be a large constant. Conjecture~\ref{conj:plant} asserts that there is a constant $c>0$ such that, for every such $H$, circuits of size at most $\exp(n^c)$ have advantage at most $1/4$ in distinguishing the resulting graph from the random background alone, even given $H$. This is a variant of the usual random-plant model~\cite{BCCFV10,dhawan2025detection}. Here the plant is fixed and known, and expansion restricts its structure.

Second, Conjecture~\ref{conj:xor} asserts the same $1/4$ advantage bound for noisy $k$XOR on a \emph{random} matrix with $N$ columns, $\lfloor N^{\delta k}\rfloor$ independent rows of weight $k$, and noise rate $1/2-\delta$, for some constant $\delta>0$. This problem belongs to the sparse LPN family studied in cryptography since Alekhnovich's work~\cite{alekhnovich2003more,ishai2008overhead,applebaum10,dao2023sparse,ghosal2025using}. Prior work gives search-to-decision reductions for noisy sparse parities~\cite{bogdanov2019xor,bogdanov2025sample} and lower bounds for statistical queries on planted equations and sum-of-squares refutation on random constraints~\cite{feldman2015complexity,kothari2017sum}.

\begin{theorem}[Informal version of Theorem~\ref{thm:pke}]\label{thm:pke-intro}
Assuming Conjecture~\ref{conj:xor} and Conjecture~\ref{conj:plant}, there is a polynomial-time public key encryption scheme that is semantically secure~\cite{goldwasser1984probabilistic} against circuits of size $\exp(n^c)$ for some constant $c>0$.
\end{theorem}

\paragraph{Encryption from Goldreich's Generator.}
There is also a Goldreich-based version of the scheme. Use the expanding instance from Corollary~\ref{cor:xor-maj} as the hidden object, and encrypt zero by evaluating the public XOR-MAJ generator on a uniform seed. Encrypt one by sampling a uniform vector. The secret key again identifies the hidden coordinates. The 
distinguisher in Corollary~\ref{cor:xor-maj} allows the receiver to recognize a generator evaluation with high probability. On a \emph{random} hypergraph, Conjecture~\ref{conj:goldreich} asserts that there is some $c>0$ such that circuits of size at most $\exp(n^c)$, given the graph, have advantage at most $1/4$ in distinguishing one such evaluation from a uniform word. The arities and stretch lie within the challenge of Applebaum and Lovett, who prove subexponential lower bounds for algebraic refutation and inversion on these random generators~\cite[Theorem~5.5 and Corollary~5.7]{applebaum2016algebraic}.

For this application, the planting conjecture must record the two roles of a predicate's inputs. A marked hyperedge specifies both its vertices and which half feed the parity gate; the other half feed the majority gate. We plant this information together with the hidden hypergraph, among random marked edges. The Goldreich scheme uses the marked planting conjecture in Conjecture~\ref{conj:marked-plant}.

\begin{theorem}[Informal version of Theorem~\ref{thm:goldreich-pke}]\label{thm:goldreich-pke-intro}
Assuming Conjecture~\ref{conj:marked-plant} and Conjecture~\ref{conj:goldreich}, there is a polynomial-time public key encryption scheme that is semantically secure against circuits of size $\exp(n^c)$ for some constant $c>0$.
\end{theorem}

\paragraph{Evidence for the Planting Conjectures.}
The log-density prediction compares the frequency of small patterns in a random plant with their frequency in the background~\cite{BCCFV10,dhawan2025detection}. In our range $n=o(\sqrt N)$, compare the edge counts on their respective vertex scales. A random plant at our budget has order $n^{\varepsilon k}$ edges on $n$ vertices, while the background has order $N^{C_2\varepsilon k}$ edges on $N$ vertices. The strict inequality $\varepsilon k<C_2\varepsilon k$ puts this benchmark on the hard side of the log-density threshold~\cite{dhawan2025detection}. Our plant can instead be fixed and known. Expansion controls how many vertices two copies must share for their edges to overlap. We use this to show that every real polynomial of degree at most $\lfloor n^{1-\varepsilon}\rfloor$ in the edge indicators changes its mean by $o(1)$ null standard deviations, uniformly over admissible plants (Section~\ref{sec:lowdegree}). The low-degree heuristic interprets this as evidence for subexponential hardness~\cite{kunisky2019notes}.

In both planting experiments the public incidence matrix has one row per edge and one column per vertex (per vertex and role for marked edges). Its columns generate a binary code $C$. Several attacks on McEliece variants measure algebraic properties of such a code, sometimes after deleting coordinates. Puncturing simply deletes those coordinates. Shortening first requires the words to vanish there. At our scheme parameters we obtain the following bounds.
\begin{enumerate}
\item \emph{Schur powers.} Square-code distinguishers for disguised Reed--Solomon and wild Goppa codes measure the dimension spanned by coordinatewise products of two words, often after shortening~\cite{CGGOT14,COT14}. For our code, with high probability the public edge count determines this dimension for products of \emph{every} order. We also give bounds after shortening and for the multiplier spaces used in closure tests~\cite{couvreur2014subcodes} (Section~\ref{sec:schur}).
\item \emph{Equations on generator rows.} Goppa distinguishers use low-rank quadratic equations on dual rows or count successive linear relations among equations on a shortened dual~\cite{couvreurMoraTillichQuadrics,randriambololonaSyzygy}. In our setting, with high probability, every equation of fixed degree below the background edge-count exponent that holds on the public incidence rows already holds on every possible edge. On a shortened dual, the retained length and rank determine an initial sequence of relation counts whenever sufficiently many disjoint sets of primal rows still span (Section~\ref{sec:incidence-syzygies}).
\item \emph{Hermite interpolation.} Ghoshal, Ishai, Jain, and Sun~\cite{GIJS26} seek a degree-$d$ homogeneous multilinear polynomial equal to one at a held-out row and vanishing to order $d-1$ at the others. This means that its expansion about each of those rows has no term of degree below $d-1$. Choose coordinates on the span of our retained incidence rows. With high probability the only degree-$d$ homogeneous form with this vanishing in those coordinates is zero, even allowing repeated variables. This holds for $d\ge4$ without marks and $d\ge7$ with marks under the withholding bound of Theorem~\ref{thm:hermite-rows}. On the dual, related tests use higher-order vanishing~\cite{hemmertWiemersVanishing}. We bound the distinguishing advantage of interpolation feasibility at every independent retention rate (Corollary~\ref{cor:random-dual-hermite}).
\item \emph{Dual squares and hulls.} The high-rate Goppa distinguisher of Faug{\`e}re et al.\ admits an explanation through the dimension of the dual square~\cite{FGO13,moraTillichDualSquare}, and Sendrier's support-splitting algorithm uses weights in the hull (the intersection of a code with its dual) after puncturing~\cite{sendrierSupportSplitting}. In our experiments the dual square has full dimension with high probability. The full hull dimension has statistically close laws, as does the array of parities counting edges that contain up to four specified incidences (Section~\ref{sec:dual-hull}).
\end{enumerate}

Expansion and the number of random edges enter these bounds in different places. Every small set of hidden rows has a vertex occurring just once, so its rows are independent. This prevents a local rank deficit and lets us select individual coordinates in high Schur powers. At low degrees the random edges already determine every polynomial function on the complete edge set. There are also enough random edges to span the incidence space many times over, giving dual words whose supports meet only at a prescribed coordinate.
\paragraph{Earlier Encryption Schemes.}
A cousin of the noisy $k$XOR problem is the learning parity with noise (LPN) problem; the main difference is that the matrix $\mathbf{G}$ is random and \emph{dense}. The conjectured hardness of LPN with constant noise rate implies efficient symmetric encryption~\cite{applebaum2009fast}. For public key encryption, Alekhnovich~\cite{alekhnovich2003more} used LPN with inverse-polynomial noise rate, and Applebaum, Barak, and Wigderson~\cite[Theorems~2.2 and~2.3]{applebaum10} obtained schemes from noisy $k$XOR, again with vanishing noise rate in their super-quasipolynomial security regimes. Ben-Sasson et al.~\cite{ben2016public} studied the noise rates compatible with decryption by an inner product.

Yu and Zhang~\cite[Sections~4.2--4.4]{yu2016cryptography} obtained PKE from a sufficiently strong subexponential hardness conjecture on \emph{constant}-noise LPN. The decryption key of their chosen-plaintext secure scheme has logarithmic Hamming weight. Agrawal, Saha, Schwartzbach, Vanukuri, and Vasudevan~\cite[Section~6]{agrawalSparseSum} also obtain constant-noise instantiations, using a short set of public vectors whose sum is zero. In their constant-noise regime this set has at most logarithmic size. In both cases there are at most $n^{O(\log n)}$ possible supports in a public instance of polynomial size in $n$. Thus these earlier noisy-linear schemes either use a vanishing noise rate or admit quasipolynomial search at constant noise. With our expanding-hypergraph planting conjecture, we get plausible security against circuits of size $\exp(n^c)$ for some $c>0$ while keeping the noise rate constant. The secret can have polynomially many coordinates because testing a noisy codeword on those coordinates takes polynomial time.

Public key encryption from random local generators and planted graphs goes back to Applebaum, Barak, and Wigderson~\cite[Section~8]{applebaum10}. Their Goldreich-based construction hides a set of $\Theta(\log n)$ output coordinates and decrypts by exhaustive search on its neighboring inputs. Bogdanov, Kothari, and Rosen~\cite{bogdanov2023public} later used planted hyperloops, small collections of parity constraints whose sum is zero, together with the $\mathrm{XOR}_3\oplus\mathrm{AND}_2$ predicate. Both constructions admit quasipolynomial-time searches for their logarithmic-size hidden structures, meaning that the security level is at best quasi-polynomial. Our hidden instance has polynomial size, and the decryption algorithm still runs in polynomial time. This gives the first PKE construction from a \emph{constant-locality} Goldreich generator and a planted-hypergraph conjecture that has plausibly subexponential security. Earlier noisy-linear constructions also admit Goldreich formulations, but their locality grows with the security parameter~\cite[Section~1.6]{bogdanov2023public}.

\paragraph{On the ``Win-Win'' Perspective for Hardness Conjectures for Cryptography.}
One of the great advances in  cryptography that occurred with the advent of probabilistic encryption~\cite{goldwasser1984probabilistic} and related breakthroughs is what is known as the ``\emph{win-win}'' perspective with regard to hardness conjectures (see, e.g.,~\cite{goldwasser2015cryptographic} and references therein). This posits that, ideally, hardness conjectures that are useful for cryptography should be longstanding  conjectures in mathematics, like the hardness of integer factorization. When this is so, then either we get a secure cryptosystem, or we refute a longstanding conjecture. This perspective is especially useful from the point of view of designing cryptosystems with an eye toward real-world deployment. 

But what about theoretical investigations into the foundations of cryptography, like the present research? Indeed, if we look historically, we find that many longstanding mathematical conjectures central to cryptography actually \emph{arose} from cryptography, perhaps most notably the Diffie Hellman problem, where the Diffie Hellman problem was based on the longstanding conjecture that Discrete Logarithm is hard. Similarly, the Decisional Unbalanced Expansion (DUE) Assumption from~\cite{applebaum10} had not been posed earlier, and was far from a longstanding conjecture in mathematics, but it too was inspired by the longstanding conjecture that the Densest Subgraph Problem is hard on average. Recall that this novel assumption, which has not been broken since, was combined with $k$XOR with an inverse polynomial error rate to yield a plausibly subexponentially secure PKE scheme.

Nevertheless, we argue that there is a strong ``win-win'' flavor to these hardness conjectures, including our new planted hypergraph conjectures, which is inspired by a long line of work on planted inference. By posing these hardness conjectures that are proven to be useful cryptographically, we create a scientific environment where \emph{breaking these conjectures would provide critically needed understanding}
about what freedom we have in co-designing cryptographic schemes and the hardness conjectures they rely on. 

Indeed, this is not speculation: the  goal of building indistinguishability obfuscation schemes was achieved \cite{jain2021indistinguishability} only through a nearly decade-long research community process of exploring many hardness conjectures that turned out to be broken (e.g. annihilation attacks~\cite{miles2016annihilation} breaking the hardness conjecture underlying the first cryptographic iO scheme~\cite{garg2016candidate}). In our context, one win becomes ``deepening our understanding of the hardness of natural average-case problems'', while the other stays ``having a secure cryptosystem.''
Recent evidence suggests how useful it can be to posit general recipes for hardness like the low degree conjecture of Hopkins~\cite{
hopkins2018statistical} even when the conjecture is subject to attacks~\cite{buhai2025quasi}. 

We believe that adopting this wider ``win-win'' perspective, already implicitly advocated by~\cite{applebaum10}, is crucial to developing deeper and more diverse foundations for cryptography. 

\paragraph{Notation.} We write $[n]=\{1,\ldots,n\}$ and $\exp(x)=e^x$. We use $U_m$ to denote the uniform distribution over vectors in $\F_2^m$ and $\mathbf{u}_m$ to denote a vector drawn from $U_m$.

\section*{Acknowledgements and AI Use Methodology}
\begin{flushleft}
\renewcommand{\arraystretch}{1.3}
\begin{tabular}{@{}l@{\qquad}c@{}}
\textbf{Activity} & \textbf{AI used} \\
Asking the research question & [No] \\
Coming up with the approach & [Yes] \\
Proof development & [Yes] \\
Writing and exposition & [Yes] \\
Checking for bugs & [Yes] \\
Other supporting tasks & [Yes] \\
\end{tabular}
\end{flushleft}

This research project began when the human authors observed that a constructive version of the counterexamples to Goldreich's PRG by Oliveira, Santhanam, and Tell~\cite{oliveira2018expander} would give a constructive counterexample to the noisy $k$XOR conjecture over arbitrary expanders (albeit with quasipolynomial stretch). Using the Reed--Muller decoder by Saptharishi, Shpilka, and Volk~\cite{SSV} along with some techniques for constructing sparse graphs from collections of low-degree polynomials,\footnote{We later learned that similar techniques were used by Applebaum and Kachlon~\cite{applebaum2019sampling}.} we managed to find the required construction (concurrently with~\cite{BLRS26}). An initial version of this work\footnote{Titled: ``Public Key Encryption from High-Corruption Constraint Satisfaction Problems.''} gave the construction along with the implied public key encryption scheme, which assumed hardness of quasipolynomial-stretch $k$XOR plus a version of Conjectures~\ref{conj:plant}~and~\ref{conj:marked-plant} stated in terms of large alphabet random predicate CSPs (LARP-CSPs). The initial version also gave lower bounds against low degree algorithms in support of this conjecture, and the authors found (but did not publish) lower bounds against a broad family of Schur product attacks.

We then began a search with Codex / ChatGPT 6 Astra Ultra using harness components from the UCLA Moonshot Harness Project~\cite{MoonshotHarness} to find a \emph{polynomial-stretch, constant-arity} counterexample to the noisy $k$XOR conjecture over arbitrary expanders. We gave our original Reed-Muller-based construction (along with the constructions by~\cite{BLRS26}) as a starting point and suggested that Codex extend the techniques beyond Reed-Muller codes. Codex eventually found the desired construction via generalized Reed--Muller codes and provided a draft of Section~\ref{sec:codes}. We digested and simplified the draft, and we discovered the toy construction described in Section~\ref{sec:warmup}. Then, with some help from Codex, we combined the new constructions with our initial work, added the corollary for Goldreich's PRG, and added the PKE scheme that uses Goldreich's PRG in place of $k$XOR. We also used Codex / ChatGPT 6.2 Astra Max to design figures and to find additional lower bounds in support of Conjectures~\ref{conj:plant}~and~\ref{conj:marked-plant}, which were carefully verified and digested by the authors, and then incorporated into the present work (see Section~\ref{sec:interpolation}) with the help of Codex.

Of course, the human authors take full responsibility for this paper and its contents. If any reader is aware of references that should be added, please contact us and we will update the manuscript. The present work subsumes our initial version.

This research was supported in part by a Laude Moonshot seed award, a Simons Investigator Award, a DARPA expMath award, NSF grant 2333935, BSF grant 2022370, a Xerox Faculty Research Award, a Google Faculty Research Award, an Okawa Foundation Research Grant, and the Symantec Chair of Computer Science.

\section{A Warm-Up with Random Lines in the Plane}\label{sec:warmup}

Our goal in this section is to see how a matrix can expand well on small sets of rows even though we can recognize noisy sums of its columns. We take the points of an affine plane over a finite field as the rows and the indicators of random lines as the columns. Two distinct rows can have a one in the same column only if we sample the line through their two points. This will give us expansion. To recognize a noisy codeword, we will use the linear equation defining each line.

Take $q$ to be a power of two, let $k=\lceil C\log q\rceil$ for a sufficiently large absolute constant $C$, and sample $n=qk$ affine lines in $\F_q^2$ independently and uniformly. The $j$th column of $\mathbf{G}\in\F_2^{m\times n}$ is the indicator of the $j$th sampled line, with the $m=q^2$ points in a fixed order. Figure~\ref{fig:line-columns} shows three such columns. Each sampled line contributes its own column, even if two draws give the same line. A uniform line contains a given point with probability $1/q$, so $k=\Theta(\log n)$ is the \emph{expected} row weight. The row weights can vary, and the stretch is $m/n=q/k=\Theta(n/(\log n)^2)$.

\begin{figure}[H]
\centering
\begin{tikzpicture}[x=1cm,y=1cm,every node/.style={font=\small}]
  \draw[thick] (-0.42,0) -- (2.12,0) node[right] {$L_1$};
  \draw[thick,dashed] (-0.42,1.3) -- (2.12,1.3) node[right] {$L_2$};
  \draw[blue!65!black,very thick] (0,0) -- (1.99,1.52) node[above right] {$L_3$};
  \fill (0,0) circle (2pt) node[below left] {$00$};
  \fill (1.7,0) circle (2pt) node[below] {$10$};
  \fill (0,1.3) circle (2pt) node[above left] {$01$};
  \fill (1.7,1.3) circle (2pt) node[above left] {$11$};
  \draw[->,thick] (2.95,.65) -- (3.83,.65);
  \node[anchor=west] at (4.15,.65) {$\begin{array}{c|ccc}
       & L_1 & L_2 & \textcolor{blue!65!black}{L_3}\\ \hline
  00&1&0&\textcolor{blue!65!black}{1}\\
  10&1&0&0\\
  01&0&1&0\\
  11&0&1&\textcolor{blue!65!black}{1}
  \end{array}$};
\end{tikzpicture}
\caption{Three sampled lines in $\F_2^2$ and the corresponding incidence matrix. A row label $ab$ gives the point $(a, b)$. The line $L_3$ contains the first and last points in the displayed row order.}
\label{fig:line-columns}
\end{figure}

To argue expansion, we need to show that a small set $S$ of rows touches nearly $k|S|$ distinct columns. In Figure~\ref{fig:line-columns}, for example, the rows labeled $00$ and $11$ have four nonzero entries between them. They touch only three columns, since the line $L_3$ passes through both points. In general, we want to count each sampled column once if its line meets $S$, however many points of $S$ lie on it. This count is exactly $|\Gamma_{\mathbf{G}}(S)|$.

\begin{proposition}\label{prop:lines}
With probability $1-O(q^{-2})$, every row of $\mathbf{G}$ has weight between $0.9k$ and $1.1k$, and
\[
 |\Gamma_{\mathbf{G}}(S)|\ge 0.9k|S|\qquad\text{for every }S\subseteq\F_q^2
       \text{ with }|S|\le n^{0.9}.
\]
\end{proposition}

\begin{proof}
Fix a set $S$ of $s$ points and consider a single random line. Summing the probabilities that it contains each of the $s$ points counts a line through several points of $S$ several times. To bound the overcount, we need the probability that the line passes through a given pair. There are $q+1$ directions, each with $q$ parallel lines. Of the $q(q+1)$ lines, $q+1$ pass through a given point and just one passes through a given pair of distinct points. Thus inclusion--exclusion gives
\[
 \Pr[\text{a uniform line meets }S]\ \ge\
 \frac{s}{q}-\frac{\binom{s}{2}}{q(q+1)}
 \ \ge\ 0.99\,\frac{s}{q}
\]
for all sufficiently large $q$, since $s\le(qk)^{0.9}=o(q)$. Now recall that each of our $n=qk$ columns comes from an independent draw. The number of columns meeting $S$ is therefore binomial with mean at least $0.99ks$. By the Chernoff bound, the probability that fewer than $0.9ks$ columns meet $S$ is at most $\exp(-\Omega(ks))$. This is where the logarithmic choice of $k$ helps. Since $k\ge C\log q$, we can make this probability at most $q^{-4s}$ by taking $C$ large enough. There are at most $q^{2s}$ sets of size $s$, so summing $q^{-2s}$ over the allowed sizes $s\ge1$ gives total failure probability $O(q^{-2})$.

Taking $S$ to be a singleton gives the lower bound on row weight. For the upper bound, each row weight is binomial with mean $k$, so another Chernoff bound and a union bound over the $q^2$ points give the same failure estimate.
\end{proof}

On this same event, consider the columns that meet a set $S$ of $s\le n^{0.9}$ rows exactly once. These are the \emph{unique neighbors} of $S$. There are at most $1.1ks$ ones in the selected rows. If we set aside one of these ones in each neighboring column, at most $1.1ks-|\Gamma_{\mathbf{G}}(S)|$ ones remain. Every column shared by two or more rows of $S$ accounts for at least one of the remaining ones. Thus at least $2|\Gamma_{\mathbf{G}}(S)|-1.1ks\ge0.7ks$ columns are unique neighbors.

\paragraph{Line Indicators as Polynomial Evaluations.}
Now fix the sampled lines and think of $\mathbf{x}\in\F_2^n$ as choosing which columns to add. At a point $\mathbf{z}$, the bit $(\mathbf{G}\mathbf{x})_{\mathbf{z}}$ is the parity of the chosen lines through $\mathbf{z}$. We would like an identity that this word satisfies regardless of which lines were chosen. For this, observe that the indicator of the line $a z_1+b z_2+c=0$, written as a polynomial over $\mathbb{F}_q$, is $1-(a z_1+b z_2+c)^{q-1}$, since every nonzero field element has $(q-1)$st power one. This polynomial takes values in $\{0,1\}$, and adding these values in $\F_q$ computes their parity because $q$ is a power of two. Thus every codeword $\mathbf{G}\mathbf{x}$, when cast to $\mathbb{F}_q$, is the evaluation table of a bivariate polynomial over $\F_q$ of total degree at most $q-1$.

More generally, we write $\RM_q(d,R)$ for the evaluations on $\F_q^d$ of all polynomials of total degree at most $R$. These are the generalized Reed--Muller codes~\cite{kasami1968generalizations,DGM}, extending the binary Reed--Muller codes~\cite{muller1954application,reed1954class}. Taking only the binary words in such a code gives its \emph{binary subfield subcode}~\cite{delsarte1975subfield}. Our columns generate a code contained in the binary subfield subcode of $\RM_q(2,q-1)$. We can therefore test identities over $\F_q$ even though the received word is binary.

\begin{proposition}\label{prop:lines-distinguishing}
For every sampled $\mathbf{G}$, a deterministic polynomial-time test accepts every $\mathbf{G}\mathbf{x}+\mathbf{e}$ with $|\supp(\mathbf{e})|\le m/64$. It accepts a uniformly random binary word with probability at most $\exp(-\Omega(m))$.
\end{proposition}

\begin{proof}
First observe that any bivariate polynomial of total degree less than $2(q-1)$ sums to zero over the plane. Indeed, in each monomial at least one variable has exponent $j<q-1$. Summing over that variable gives zero because $\sum_{u\in\F_q}u^j=0$ for $0\le j<q-1$ (for $j=0$ the sum is $q=0$ in $\F_q$). Recall that a codeword is the evaluation table of a polynomial of degree at most $q-1$. Multiplying this polynomial by a monomial of degree less than $q-1$ and summing over the plane therefore gives zero.

One flipped bit can spoil many of these identities. Following the syndrome methods of~\cite{SSV,kopparty2018syndrome}, we arrange the sums in a matrix indexed by \emph{pairs} of monomials. For each pair we take the inner product of the received word $\mathbf{y}$ with the evaluation of their product. The reason to do this is that one flipped bit will change the matrix by a single outer product. If both monomials have degree at most $r$, we can apply the degree bound above to every entry provided $2r+(q-1)<2(q-1)$. Take the largest such integer, $r=(q-2)/2$.

Let the rows of $\mathbf{V}$ be indexed by the same points as the rows of $\mathbf{G}$. Its columns now correspond to the bivariate monomials of total degree at most $r$, each evaluated at all $m=q^2$ points. The $h=\binom{r+2}{2}=q(q+2)/8$ columns are linearly independent over $\F_q$. Indeed, a polynomial of individual degrees less than $q$ that vanishes on the plane is zero by the univariate root bound, applied in each variable. Compute
\[
 \mathbf{M}(\mathbf{y})=\mathbf{V}^{\mathsf T}\diag(\mathbf{y})\mathbf{V}\quad\text{over }\F_q,
 \qquad\text{and accept if }\rank \mathbf{M}(\mathbf{y})\le h/6.
\]
Write $\mathbf{v}_{\mathbf{z}}^{\mathsf T}$ for the row of $\mathbf{V}$ at point $\mathbf{z}$. Then $\mathbf{M}(\mathbf{y})=\sum_{\mathbf{z}} y_{\mathbf{z}}\mathbf{v}_{\mathbf{z}}\mathbf{v}_{\mathbf{z}}^{\mathsf T}$. On a codeword $\mathbf{w}$, each entry sums a polynomial of degree at most $2r+(q-1)=2q-3$ over the plane, so $\mathbf{M}(\mathbf{w})=\mathbf{0}$. Now flip the bit at $\mathbf{z}$. This adds $\mathbf{v}_{\mathbf{z}}\mathbf{v}_{\mathbf{z}}^{\mathsf T}$ to the matrix, regardless of how many entries change. Starting from $\mathbf{M}(\mathbf{w})=\mathbf{0}$, we add one such rank-one matrix per error, so $\rank \mathbf{M}(\mathbf{w}+\mathbf{e})\le|\supp(\mathbf{e})|$. Since $h/6\ge m/48$, the test accepts every error of weight at most $m/64$.

It remains to see why a uniform binary word $\mathbf{r}$ is unlikely to pass. Its bits are independent, but a single bit can appear in many entries of $\mathbf{M}(\mathbf{r})$. We will change basis so that for $h$ of the points, each bit appears in just one diagonal entry. To do this, recall that $\mathbf{V}$ has column rank $h$, so we can choose a set $I$ of $h$ points whose rows are linearly independent. Write $\mathbf{V}_I$ for their invertible submatrix, let $\mathbf{V}'=\mathbf{V}\mathbf{V}_I^{-1}$, and define
\[
 \mathbf{M}'(\mathbf{y})=\mathbf{V}'^{\mathsf T}\diag(\mathbf{y})\mathbf{V}'\quad\text{over }\F_q.
\]
Multiplication on the left and right by invertible matrices preserves rank, so $\mathbf{M}'(\mathbf{y})$ has the same rank as $\mathbf{M}(\mathbf{y})$. The rows of $\mathbf{V}'$ indexed by $I$ form the identity matrix. Condition on all bits of $\mathbf{r}$ outside $I$. Their outer products give a fixed matrix $\mathbf{A}$, and each remaining bit contributes only to its own diagonal entry. Thus $\mathbf{M}'(\mathbf{r})=\mathbf{A}+\diag(\mathbf{b})$ for uniform $\mathbf{b}\in\F_2^h$.

We want a bound that holds for every fixed $\mathbf{A}$. For this, compare any two choices $\mathbf{b}_0,\mathbf{b}$ that both pass the test. Subtracting the two matrices cancels $\mathbf{A}$, leaving a diagonal whose rank counts the positions where the choices differ. Since the rank of a difference is at most the sum of the two ranks, we get
\[
 |\supp(\mathbf{b}-\mathbf{b}_0)|=\rank\diag(\mathbf{b}-\mathbf{b}_0)
 \le \rank(\mathbf{A}+\diag(\mathbf{b}))+\rank(\mathbf{A}+\diag(\mathbf{b}_0))\le h/3.
\]
Thus, if any choice passes, fix one such $\mathbf{b}_0$. All accepted choices lie in the Hamming ball about $\mathbf{b}_0$ of radius $2(h/6)=h/3$. This explains the choice of rank threshold. A uniform $\mathbf{b}$ differs from this fixed center in $h/2$ coordinates on average, so the fraction that pass is at most $\Pr[\operatorname{Bin}(h,1/2)\le h/3]=\exp(-\Omega(h))$. If there is no such $\mathbf{b}_0$, the acceptance probability is zero. Since $h\ge m/8$, the bound holds with exponent $\Omega(m)$ after removing the conditioning. The test uses one matrix rank computation over $\F_q$.
\end{proof}

The slack between $m/64$ and $m/48$ also lets us handle independent $\Ber(1/64)$ errors. Their weight exceeds $m/48$ with probability $\exp(-\Omega(m))$, so the same test distinguishes noisy parities from uniform with exponentially small error on either distribution.

\paragraph{Connections to Algebraic Decoding.}
The entries of the matrix above have been used to locate errors in algebraic codes. An \emph{error-locating polynomial} is a nonzero auxiliary polynomial that vanishes at the corrupted evaluation points. Pellikaan and Duursma--K{\"o}tter developed general error-location methods using coordinatewise products of codes~\cite{pellikaan1992decoding,duursma1994error}. For Reed--Muller codes, Abbe, Shpilka, and Wigderson~\cite{ASW} related recovery from errors to recovery from erasures. Saptharishi, Shpilka, and Volk~\cite{SSV} gave an efficient reduction using spaces of error-locating polynomials, and Kopparty and Potukuchi~\cite{kopparty2018syndrome} gave faster syndrome-decoding algorithms. One of their algorithms finds these polynomials through a linear system whose coefficients are inner products of the received word with products of monomials, as in our matrix. Our analysis uses the rank of this matrix without recovering the error locations.

B{\l}asiok, Lou, Rosen, and Sudan~\cite{BLRS26} also generate codes by indicators of affine subspaces. They obtain these subspaces from the Guruswami--Umans--Vadhan expander~\cite{guruswami2009unbalanced}, based on Parvaresh--Vardy codes~\cite{parvaresh2005correcting}, and recognize noisy codewords by recovering a candidate with the binary Reed--Muller decoder of~\cite{SSV}. Their unconditional construction is explicit and has quasipolynomial stretch. Our warm-up already gives polynomial stretch and a rank test that tolerates up to $m/64$ arbitrary bit flips. The row weights are still logarithmic with high probability. Next we make every row have the same constant weight and bring the error fraction arbitrarily close to one half.

\section{The Optimized Construction with Constant Arity}\label{sec:codes}

In the plane, independent lines gave logarithmic row weight and a test tolerating an error rate of $1/64$. To make the row weight constant, consider instead all the parallel lines in one direction. They partition the plane, so every point lies on exactly one of them. In our constant-arity construction, we will use $k$ partitions of a higher (but still constant) dimensional point set, with one column for each part. The row weight is then exactly $k$. For expansion we need a small set of points to occupy almost as many parts as it has incidences.

Our points will lie in $\F_q^d$. In a given partition, first group the points according to the value of one linear form. We choose these first forms so that their values in any $d$ partitions determine the point. Within each group use two fresh random linear forms to subdivide it. Two distinct points in that group have probability $q^{-2}$ of remaining together. Each resulting part is the solution set of three affine linear equations, that is, an affine flat (possibly empty). Its indicator is a product of three factors of the kind we used for lines, one per equation. It therefore takes binary values and has degree at most $3(q-1)$ when interpreted as a polynomial over $\mathbb{F}_q$, independently of $d$.

The preliminary matrix has $kq^3$ columns, one per part, far more than we want in the final matrix. We can prove expansion up to a fixed power of this preliminary width. To obtain a cutoff close to the \emph{final} width $n$, we randomly combine parts within each partition until at most $n$ columns remain. The folding argument bounds the additional collisions on the same row sets. The expansion cutoff becomes $n^{1-\varepsilon}$, now measured against the smaller width.

Each point still lies in exactly one part per partition, and a sum of the indicators still has degree at most $3(q-1)$ when interpreted as a polynomial over $\mathbb{F}_q$. We can therefore use the same rank test as in the plane. The reason to increase the number of variables is that the space of monomials available for this test can now have dimension just below half the number of points. To show that a uniform word gives rank close to that dimension, we use an erasure bound for generalized Reed--Muller codes~\cite{KKMPSU}. This is where we improve on the warm-up's error threshold.

Nearly optimal expansion is already possible with constant degree~\cite{capalbo2002randomness} and with a large imbalance in the matrix dimensions~\cite{tashma2007lossless}. The innovation in our construction is to achieve that expansion while every generator column stays in the binary subfield subcode of $\RM_q(d,3(q-1))$.

\subsection{The Construction}

We prove the following slightly more precise form of Theorem~\ref{thm:main}.

\begin{theorem}\label{thm:codes}
For every constant $\varepsilon\in(0,1/2)$, there are integers $d=O(\varepsilon^{-2})$ and $k=O(\varepsilon^{-6})$ with the following properties. For every sufficiently large $n$, a randomized $\operatorname{poly}_{\varepsilon}(n)$-time algorithm samples a matrix $\mathbf{G}\in\F_2^{m\times n}$, where $m=n^{2d^2+o(1)}<n^{\varepsilon k}$. Every output has row weight $k$, and with probability at least $1-n^{-1/\varepsilon}$ the matrix $\mathbf{G}$ is $(1-\varepsilon,k,n^{1-\varepsilon})$-expanding. There is a deterministic $\operatorname{poly}_{\varepsilon}(m,n)$-time test $D$ such that, for every output $\mathbf{G}$,
\[
 D(\mathbf{c}+\mathbf{e})=1\quad\text{for every }\mathbf{c}\in\operatorname{im}_{\F_2}\mathbf{G},\ \mathbf{e}\in\F_2^m
 \text{ with }|\operatorname{supp}(\mathbf{e})|\le(1/2-\varepsilon)m,
\]
and $\Pr[D(\mathbf{u}_m)=1]\le\exp(-\Omega_{\varepsilon}(m))$.
The test uses the specified ordering of the rows, but no randomness used to sample $\mathbf{G}$.
\end{theorem}

Fix $\varepsilon$. It suffices to give the construction for rational $\varepsilon$, since any fixed rational accuracy in $(\varepsilon/2,\varepsilon)$ gives stronger expansion, noise tolerance, and success probability, with the stated orders of growth for $d$ and $k$. Our parameters serve two purposes. Taking $d$ sufficiently large makes the generalized Reed--Muller code used by the rank test have dimension close to half its length. Taking $k$ sufficiently large compared with $d$ allows us to lose some neighbors in the affine construction and in the subsequent column-combining step while retaining $(1-\varepsilon)$ expansion. The estimates below are satisfied by the choices
\begin{equation}\label{eq:code-parameters}
 d=2\left\lceil512\varepsilon^{-2}\right\rceil,
 \qquad k=\left\lceil4d^2\varepsilon^{-2}\right\rceil.
\end{equation}
We take the field size large enough that the random collisions in the expansion proof are unlikely: let $q$ be the least power of two at least $n^{2d}$. After constructing the affine partitions, we will combine their columns into $k$ groups of $b$ columns each, where we set $b$ to be the largest power of two at most $n/k$. We index the $m=q^d$ rows by $\F_q^d$ in a fixed order. Then
\begin{equation}\label{eq:code-stretch}
 n/(2k)<b\le n/k,
 \qquad n^{2d^2}\le m<2^d n^{2d^2}<n^{\varepsilon k}.
\end{equation}
Here and below, $n$ is sufficiently large as a function of the fixed constant $\varepsilon$.

\paragraph{The Affine Partitions.}
We first construct the $k$ affine partitions. We require the first equation's values in any $d$ layers to determine the row index uniquely; this will control row sets that have collisions in many layers. Take distinct $\alpha_1,\ldots,\alpha_k\in\F_q$ and define
\[
 \lambda_i(\mathbf{z})=\sum_{j=0}^{d-1}\alpha_i^jz_j.
\]
Any $d$ distinct $\lambda_i$ are linearly independent, by the Vandermonde determinant. 

Within each level set $\lambda_i(\mathbf{z})=a$, we use two further equations to divide the rows at random. For each $(i,a)\in[k]\times\F_q$, independently sample two uniform linear forms $\mu_{i,a},\nu_{i,a}:\F_q^d\to\F_q$. Sampling separately for every level set makes the randomness used in different level sets independent, as required by the expansion proof. Define a preliminary matrix $\mathbf{A}$, with columns indexed by $[k]\times\F_q^3$,~by
\begin{equation}\label{eq:code-affine}
 \mathbf{A}_{\mathbf{z},(i,a,u,v)}=
 \mathbf1[\lambda_i(\mathbf{z})=a,\ \mu_{i,a}(\mathbf{z})=u,\ \nu_{i,a}(\mathbf{z})=v].
\end{equation}
We refer to the columns with first index $i$ as layer $i$. Each row has exactly one neighbor in each layer. Every column is the indicator of an affine flat of codimension at most three, or is zero.

\paragraph{Folding the Columns.}
The preliminary matrix has $kq^3$ columns. To reduce this number, we assign each column to one of $b$ positions within its layer and add columns assigned to the same position. We refer to this step as \emph{column folding}. For each layer $i$, independently sample a uniform function $h_i:\F_q^3\to[b]$; that is, sample all its values independently. Combine the columns as follows:
\begin{equation}\label{eq:code-fold}
 \mathbf{G}_{\mathbf{z},(i,j)}=\sum_{(a,u,v):\,h_i(a,u,v)=j}\mathbf{A}_{\mathbf{z},(i,a,u,v)}
 \quad\text{over }\F_2,
 \qquad (i,j)\in[k]\times[b].
\end{equation}
Append $n-kb$ zero columns. Since each row had only one nonzero entry in each layer, combining columns causes no cancellation within that row. Its weight is still exactly $k$. Also, every column of $\mathbf{G}$ is a binary sum of columns of $\mathbf{A}$, and hence $\operatorname{im}_{\F_2}\mathbf{G}\subseteq\operatorname{im}_{\F_2}\mathbf{A}$. The sampler uses $kq^3$ independent hash values and $2kq$ random linear forms. Since $d$ and $k$ are constants, sampling and writing the matrix take polynomial time. Both $q$ and $b$ are powers of two, so these choices use a polynomial number of unbiased random bits.

\subsection{Proving Expansion}

We first prove expansion for $\mathbf{A}$, and then bound the loss when its columns are combined. The first argument uses the independence of the random forms in different parts of the affine partitions. The second is an occupancy bound of the kind used to prove expansion of random graphs~\cite{HLW}.

\begin{lemma}\label{lem:code-base}
Except with probability $o(n^{-1/\varepsilon})$, every set $S$ of at most $n^{1-\varepsilon}$ rows satisfies $|\Gamma_{\mathbf{A}}(S)|\ge(k-d+1)|S|$.
\end{lemma}
\begin{proof}
For $\mathbf{z}\in S$, let $\kappa_S(\mathbf{z})$ count the layers in which $\mathbf{z}$ shares its neighbor with another row of $S$. Define the neighbor deficit $\Delta(S)=k|S|-|\Gamma_{\mathbf{A}}(S)|$. Removing $\mathbf{z}$ changes this deficit by
\[
 \Delta(S)-\Delta(S\setminus\{\mathbf{z}\})=\kappa_S(\mathbf{z}).
\]
Consequently, if every nonempty small row set contains a row with $\kappa_S(\mathbf{z})\le d-1$, repeatedly removing such a row proves the desired bound. It remains to exclude a small nonempty set $S$ with $\kappa_S(\mathbf{z})\ge d$ for every $\mathbf{z}\in S$.

Suppose such an $S$ exists. We will count the level sets in which
collisions occur before counting the rows themselves. Call a pair
$(i,u)$ active if there are distinct $\mathbf{z},\mathbf{z}'\in S$ such that
\[
 \lambda_i(\mathbf{z})=\lambda_i(\mathbf{z}')=u,\qquad
 \mu_{i,u}(\mathbf{z})=\mu_{i,u}(\mathbf{z}'),\qquad
 \nu_{i,u}(\mathbf{z})=\nu_{i,u}(\mathbf{z}').
\]
Let $\mathcal A$ be the set of active pairs and let
$a=|\mathcal A|$. Each active pair uses at least two of the
$k|S|$ row-layer incidences, so $a\le k|S|/2$.

Once $\mathcal A$ is fixed, every row of $S$ must lie in
\[
 W_{\mathcal A}=\{\mathbf{z}\in\F_q^d:(i,\lambda_i(\mathbf{z}))\in\mathcal A
 \text{ for at least }d\text{ indices }i\}.
\]
Indeed, every row of $S$ collides in at least $d$ layers.
A point in $W_{\mathcal A}$ is determined by choosing any $d$
corresponding active pairs: their $d$ fixed equations have at
most one solution, by the Vandermonde property. There are at most
$a^d$ such choices, so $|W_{\mathcal A}|\le a^d$.

For each active pair, choose one ordered pair of distinct points
witnessing its collision. Since all these points belong to
$W_{\mathcal A}$, there are at most $a^{2da}$ possible choices
of all the witnesses. For a fixed choice, a uniform linear
form agrees on the two distinct points with probability
$1/q$, since it must vanish on their nonzero difference. The two random forms therefore satisfy the chosen
collision with probability $q^{-2}$. The forms used at
different active pairs are independent, even if the chosen point
pairs overlap, so the probability of all required collisions is
$q^{-2a}$. Crucially, $W_{\mathcal A}$ was defined using only
the fixed forms $\lambda_i$; choosing the possible witnesses
has not exposed any of this randomness.

There are $\binom{kq}{a}$ possible sets of active pairs. This explains the use of two random forms per level set. Allowing all $q$ choices for each level costs at most $q^a$, leaving a factor $q^{-a}$ from the collision probability for the other choices. We can
therefore union-bound first over these sets and then over their
collision witnesses:
\begin{align*}
 \Pr[\text{such a set }S\text{ exists}]
 &\le\sum_{1\le a\le kn^{1-\varepsilon}/2}
 \binom{kq}{a}a^{2da}q^{-2a}\\
 &\le\sum_{1\le a\le kn^{1-\varepsilon}/2}
 \left(\frac{ek a^{2d-1}}q\right)^a\\
 &\le\sum_{a\ge1}\left(C_{\varepsilon}n^{-1-(2d-1)\varepsilon}\right)^a
 =o(n^{-1/\varepsilon}).
\end{align*}
Here $C_{\varepsilon}$ is a constant depending only on $\varepsilon$. The last line uses $q\ge n^{2d}$ and $1+(2d-1)\varepsilon>1/\varepsilon$.
\end{proof}

\begin{lemma}\label{lem:code-folding}
Condition on any $\mathbf{A}$ satisfying Lemma~\ref{lem:code-base}. With probability $1-o(n^{-1/\varepsilon})$ over the functions $h_i$, the matrix $\mathbf{G}$ is $(1-\varepsilon,k,n^{1-\varepsilon})$-expanding.
\end{lemma}
\begin{proof}
Fix $S$ with $s=|S|\le n^{1-\varepsilon}$. Expose the hash values of its distinct neighbors in $\mathbf{A}$, one layer at a time. There are at most $s$ neighbors per layer, so each exposure repeats an already occupied value with conditional probability at most $s/b<2ks/n$. The number of repeats is exactly $|\Gamma_{\mathbf{A}}(S)|-|\Gamma_{\mathbf{G}}(S)|$. If $|\Gamma_{\mathbf{G}}(S)|<(1-\varepsilon)ks$, there must be more than $(\varepsilon k-d+1)s$ repeats among at most $ks$ exposures. A union bound over their positions gives
\[
 \Pr[S\text{ fails expansion}]
 \le2^{ks}\left(\frac{2ks}{n}\right)^{(\varepsilon k-d+1)s}.
\]
The estimate uses independent hash values for distinct preliminary neighbors; rows that already share a neighbor are not exposed separately. Single rows always have weight $k$. Taking a union bound over the other row sets and using~\eqref{eq:code-stretch}, we get
\begin{align*}
 \Pr[\text{expansion fails}]
 &\le\sum_{2\le s\le n^{1-\varepsilon}}
 m^s2^{ks}(2kn^{-\varepsilon})^{(\varepsilon k-d+1)s}\\
 &\le\sum_{s\ge2}
 \left(C_{\varepsilon}n^{2d^2-\varepsilon(\varepsilon k-d+1)}\right)^s
 =o(n^{-1/\varepsilon}).
\end{align*}
Indeed,~\eqref{eq:code-parameters} implies
\[
 \varepsilon(\varepsilon k-d+1)-2d^2
 \ge2d^2-\varepsilon d+\varepsilon>1/\varepsilon+1.
\]
\end{proof}

The probability that $\mathbf{G}$ fails to expand is at most the sum of the two failure probabilities. This is $o(n^{-1/\varepsilon})$, and hence at most $n^{-1/\varepsilon}$ for sufficiently large $n$.

\subsection{Distinguishing Corrupted Codewords}

Fix any output $\mathbf{G}$ of the sampler. We adapt the rank test from Proposition~\ref{prop:lines-distinguishing} to the points of $\F_q^d$. Each folded column is a sum of flat indicators of degree at most $3(q-1)$. The same field-sum argument as in the plane proves the orthogonality we need~\cite{DGM}. Multiply a column polynomial by two monomials of degree at most $R$. If $2R+3(q-1)<d(q-1)$, each resulting monomial has an exponent below $q-1$ in some variable, so its sum over the whole grid is zero. Take the largest integer $R$ allowed by this inequality, and put
\begin{equation}\label{eq:code-features}
 R=\left\lfloor\frac{(d-3)(q-1)-1}{2}\right\rfloor,
 \qquad B=\operatorname{RM}_q(d,R),
\end{equation}
and let $\mathbf{V}\in\F_q^{m\times r}$ be the matrix whose columns are the evaluations of the monomials of total degree at most $R$ and individual degrees less than $q$ (a basis of $B$), where $r=\dim B$. All ranks and dimensions in this subsection are over $\F_q$. For $\mathbf{y}\in\F_2^m$, set
\[
 \boldsymbol{\mathcal{M}}(\mathbf{y})=\mathbf{V}^{\mathsf T}\operatorname{diag}(\mathbf{y})\mathbf{V},
 \qquad
 D(\mathbf{y})=\mathbf1\!\left[\operatorname{rank}_{\F_q}\boldsymbol{\mathcal{M}}(\mathbf{y})
 \le(1/2-\varepsilon/2)m\right].
\]
Computing the test requires one matrix rank over $\F_q$.

\begin{lemma}\label{lem:code-annihilation}
Every $\mathbf{c}\in\operatorname{im}_{\F_2}\mathbf{G}$ satisfies $\boldsymbol{\mathcal{M}}(\mathbf{c})=\mathbf{0}$. Consequently, $\operatorname{rank}\boldsymbol{\mathcal{M}}(\mathbf{c}+\mathbf{e})\le|\operatorname{supp}(\mathbf{e})|$ for every binary vector $\mathbf{e}$.
\end{lemma}
\begin{proof}
Each entry of $\boldsymbol{\mathcal{M}}(\mathbf{c})$ pairs $\mathbf{c}$ with a product of two basis monomials. The degree choice in~\eqref{eq:code-features} makes every such pairing zero on each generator column, hence on $\mathbf{c}$. Writing the row of $\mathbf{V}$ at $\mathbf{z}$ as $\mathbf{v}_{\mathbf{z}}^{\mathsf T}$, we get
\[
 \boldsymbol{\mathcal{M}}(\mathbf{c}+\mathbf{e})=\sum_{\mathbf{z}\in\operatorname{supp}(\mathbf{e})}\mathbf{v}_{\mathbf{z}} \mathbf{v}_{\mathbf{z}}^{\mathsf T}
\]
as a sum of $|\operatorname{supp}(\mathbf{e})|$ matrices of rank at most one.
\end{proof}

It remains to understand the rank for a uniformly random binary word $\mathbf{y}$. Its support $S$ contains each coordinate independently with probability $1/2$, and $\boldsymbol{\mathcal{M}}(\mathbf{y})=\mathbf{V}_S^{\mathsf T}\mathbf{V}_S$, where $\mathbf{V}_S$ denotes the submatrix of $\mathbf{V}$ with rows indexed by $S$. Sylvester's rank inequality gives a lower bound of $2\operatorname{rank}\mathbf{V}_S-|S|$. Thus we want restriction to $S$ to retain almost all $r$ dimensions of $B$, and we want $r$ to be close to $m/2$.

The dimensions lost on restriction have an interpretation as erasure ambiguity. For a linear code $C\subseteq\F_q^M$ and $E\subseteq[M]$, write $C_E=\{\mathbf{c}\in C:\operatorname{supp}(\mathbf{c})\subseteq E\}$. Two codewords agree on the unerased coordinates exactly when their difference belongs to $C_E$, so $\dim C_E$ measures the ambiguity left by erasing $E$. In particular, the kernel of restriction from $B$ to $S$ is $B_{S^c}$. We only need its expected dimension to be $o(m)$.

The next lemma is a consequence of the erasure-capacity theorem of Kudekar, Kumar, Mondelli, Pfister, {\c{S}a\c{s}o\u{g}lu}, and Urbanke~\cite{KKMPSU}. A coordinate-permutation group is doubly transitive if it can send any ordered pair of distinct coordinates to any other such pair. For a sequence of such codes whose rates tend to $\rho$, their theorem says that the expected fraction of symbols not determined by the unerased coordinates tends to zero whenever the erasure probability is less than $1-\rho$. We use only this symbol-recovery conclusion; the erased coordinates need not all be recoverable simultaneously.

\begin{lemma}\label{lem:code-erasure}
Let $C_j\subseteq\F_{q_j}^{M_j}$ be linear codes invariant under doubly transitive groups of coordinate permutations, with $M_j\to\infty$ and rates tending to $\rho\in(0,1)$. For fixed $p<1-\rho$, let $E$ include each coordinate independently with probability $p$. Then $\mathbb E\dim(C_j)_E=o(M_j)$.
\end{lemma}
\begin{proof}
Apply the symbol-recovery theorem~\cite[Theorem~20 and Section~VI-D]{KKMPSU}. Its proof applies uniformly when $q_j$ varies: the erasure pattern is a vector of Boolean indicators, entropy is measured in base $q_j$, and the quantitative threshold bound in~\cite[equation~(8)]{KKMPSU} is independent of the alphabet. Thus the expected number of undetermined coordinates is $o(M_j)$.

A coordinate is undetermined exactly when some word of $(C_j)_E$ is nonzero there: two codewords agree on the unerased coordinates precisely when their difference lies in $(C_j)_E$. Every word of $(C_j)_E$ is therefore supported on the set of undetermined coordinates. Its dimension is at most the size of that set, giving $\mathbb E\dim(C_j)_E=o(M_j)$.
\end{proof}

We now check the two rate requirements for $B$. Its rate must be close to $1/2$ to make the rank lower bound large, and strictly below $1/2$ so that Lemma~\ref{lem:code-erasure} applies when half the coordinates are erased.

\begin{lemma}\label{lem:code-rate}
For the code $B$ in~\eqref{eq:code-features}, $r/m\to\rho_d$, where
\[
 \frac12-\frac{\varepsilon}{8}\le\rho_d
 =\Pr\!\left[\sum_{j=1}^d X_j\le\frac{d-3}{2}\right]<\frac12,
 \qquad X_j\text{ independent uniform on }[0,1].
\]
If $E$ includes each coordinate independently with probability $1/2$, then $\mathbb E\dim B_E=o(m)$.
\end{lemma}
\begin{proof}
The reduced monomials (individual exponents below $q$) form a basis. Indeed, reduction using $z_i^q=z_i$ cannot increase total degree, and these monomials are independent by the root bound in each variable. Thus $r/m$ is the probability that $d$ independent uniform exponents in $\{0,\ldots,q-1\}$ have sum at most $R$. After dividing each exponent by $q-1$, this count is a Riemann sum for the displayed probability.

The sum $\sum_jX_j$ is continuously distributed and symmetric about $d/2$. Thus $1/2-\rho_d$ is the probability of the interval $(d/2-3/2,d/2]$. We bound this probability by separating one fair bit from each $X_j$. Write $X_j=(Y_j+X'_j)/2$, where the $Y_j$ are independent fair bits and the $X'_j$ are independent uniform variables on $[0,1]$. Conditional on the $X'_j$, the event $d/2-3/2<\sum_jX_j\le d/2$ restricts $\sum_jY_j$ to an interval of length three, containing at most four integers. Since $d$ is even, the largest atom of this binomial distribution is at most $1/\sqrt{d+1}$. The bound follows, for example, by induction from $4^{-s}\binom{2s}{s}\le(2s+1)^{-1/2}$. Consequently,
\[
 0<\frac12-\rho_d\le\frac4{\sqrt{d+1}}\le\frac{\varepsilon}{8}.
\]

An invertible affine substitution preserves total degree and hence preserves $B$. These substitutions act doubly transitively on $\F_q^d$: translate one point to the origin and use an invertible linear map to send the difference of the two points to any other nonzero difference. Lemma~\ref{lem:code-erasure} applies at $p=1/2<1-\rho_d$.
\end{proof}

\begin{proposition}\label{prop:code-distinguishing}
The test $D$ satisfies the distinguishing guarantee of Theorem~\ref{thm:codes} for every sampled $\mathbf{G}$. It also accepts $\mathbf{c}+\mathbf{e}$ with probability $1-\exp(-\Omega_\varepsilon(m))$ for every $\mathbf{c}\in\operatorname{im}_{\F_2}\mathbf{G}$ when $\mathbf{e}\sim\Ber(1/2-\varepsilon)^m$.
\end{proposition}
\begin{proof}
For every $\mathbf{c}\in\operatorname{im}_{\F_2}\mathbf{G}$ and every binary error vector $\mathbf{e}$, Lemma~\ref{lem:code-annihilation} gives $\operatorname{rank}\boldsymbol{\mathcal{M}}(\mathbf{c}+\mathbf{e})\le|\operatorname{supp}(\mathbf{e})|$. Thus $D(\mathbf{c}+\mathbf{e})=1$ whenever $|\operatorname{supp}(\mathbf{e})|\le(1/2-\varepsilon)m$; in fact, the test accepts up to its larger threshold $(1/2-\varepsilon/2)m$. If $\mathbf{e}\sim\Ber(1/2-\varepsilon)^m$, a binomial upper-tail bound shows that its weight exceeds this threshold with probability at most $\exp(-\Omega_\varepsilon(m))$.

For uniform $\mathbf{y}$, put $S=\operatorname{supp}(\mathbf{y})$ and let $\mathbf{V}_S$ denote the restriction of $\mathbf{V}$ to the rows in $S$. Sylvester's rank inequality gives
\[
 \operatorname{rank}\boldsymbol{\mathcal{M}}(\mathbf{y})=\operatorname{rank}(\mathbf{V}_S^{\mathsf T}\mathbf{V}_S)
 \ge2\operatorname{rank}\mathbf{V}_S-|S|.
\]
Rank--nullity gives $\operatorname{rank}\mathbf{V}_S=r-\dim B_{S^c}$. Since $S^c$ is an independent erasure pattern of rate $1/2$, Lemma~\ref{lem:code-rate} gives
\begin{align*}
 \mathbb E\operatorname{rank}\boldsymbol{\mathcal{M}}(\mathbf{y})
 &\ge2r-2\mathbb E\dim B_{S^c}-\mathbb E|S|\\
 &=2r-m/2-o(m)\ge(1/2-\varepsilon/4-o(1))m.
\end{align*}
This expectation exceeds the test's threshold by a fixed positive multiple of $\varepsilon m$ for sufficiently large $n$. To obtain a high-probability bound, observe that flipping one bit of $\mathbf{y}$ changes $\boldsymbol{\mathcal{M}}(\mathbf{y})$ by rank at most one, and hence changes its rank by at most one. McDiarmid's bounded-differences inequality~\cite{McD} bounds the probability that the rank is at most $(1/2-\varepsilon/2)m$ by $\exp(-\Omega_{\varepsilon}(m))$. All bounds are uniform in $\mathbf{c}$ and in the sampler's output.
\end{proof}

Together with the expansion bound from Lemmas~\ref{lem:code-base} and~\ref{lem:code-folding}, Proposition~\ref{prop:code-distinguishing} proves Theorem~\ref{thm:codes}.

\paragraph{Interpolation and Erasure Decoding.}
Other decoding approaches use interpolation to obtain a polynomial relation satisfied by the received values, and then recover candidate messages from that relation. Heydtmann and Jakobsen~\cite{heydtmann2001decoding} studied this approach for binary Reed--Muller codes; their random-error analysis assumes an approximately binomial weight distribution. Pellikaan and Wu~\cite{pellikaan2004list} gave list-decoding algorithms for generalized Reed--Muller codes, which return a short list of possible messages. One of their methods reduces to Reed--Solomon codes, the evaluation codes of univariate polynomials, over an extension field. Augot and Stepanov~\cite{augot2009note} refined the multivariate interpolation analysis. 

The dimension of polynomial spaces vanishing on a random set is also central to the binary erasure bounds of Bhandari, Harsha, Saptharishi, and Srinivasan~\cite{bhandari2022vanishing}. The broader random-noise theory includes bounds on the number of codewords of each Hamming weight~\cite{sberlo2020performance}, and recovery up to the information-theoretic limits of symmetric binary-input noise channels, first for each bit~\cite{reeves2023reed} and then for the entire codeword~\cite{abbe2023proof}. These binary results concern different parameter regimes. The symbol-erasure theorem used above applies uniformly to our growing alphabet and supplies the expected-dimension bound needed by the rank test.

\subsection{Goldreich's Generator with XOR-MAJ}\label{sec:xor-maj}

We first construct one block of the generator. Each row of our matrix has one neighbor in each layer. Assigning whole layers to the parity gate ensures that the vector of parity outputs is a codeword, obtained by setting all other input coordinates to zero. We would like the remaining vector, the majority outputs, to have weight below the noise threshold of the rank test.

For this purpose we give all majority gates in the block a small common set of inputs. When those inputs happen to be zero, each majority output is biased toward zero. Conditioning on shared inputs to bias a majority also appears in the grouping attacks of Fu, Li, Lyu, and Liu~\cite{fu2026grouping}. The rank test detects the block with constant probability. Placing copies on disjoint input coordinates gives independent chances to detect the generator on a single uniform seed.

\begin{proof}[Proof of Corollary~\ref{cor:xor-maj}]
Put $\delta=\varepsilon/4$ and take $k$ and $d$ from~\eqref{eq:code-parameters} at accuracy $\delta$. Choose the smallest integer $t\ge2\sqrt{k}$ with $k+t\equiv2\pmod4$, and put
\[
 \ell=(k+t)/2,\qquad r=(k-t)/2 .
\]
We use $\ell$ parity layers and $r$ majority layers. The $t$ shared majority inputs give $\ell=r+t$ inputs to each majority gate, and $\ell$ is odd. The choice of $k$ gives $k\ge64\varepsilon^{-2}$, so $t<2\sqrt{k}+4\le\varepsilon k/2$. In particular $r>0$ and $\ell=O(\varepsilon^{-6})$.

Given $N$ input coordinates, take $B=\lfloor N^{\varepsilon/2}\rfloor$ blocks of $N_0=\lfloor N/B\rfloor$ inputs each, leaving the last $N-BN_0$ inputs unused. Sample a \emph{single} matrix $\mathbf{G}$ from Theorem~\ref{thm:codes} with $N_0-t$ columns at accuracy $\delta$. Write $m_0$ for its number of rows. Append $t$ all-one columns to $\mathbf{G}$. In each row, feed the neighbors in the first $\ell$ layers to the parity gate, and the other $r$ neighbors and the $t$ new inputs to the majority gate. This gives $2\ell$ distinct inputs per tuple. Use a copy of these same $m_0$ tuples in each block, with input indices shifted to that block. The resulting incidence matrix is block diagonal, with $m=Bm_0$ rows.

It suffices that the one sampled block expand up to $N_0^{1-\varepsilon/2}$ rows. Indeed,
\[
 (1-\varepsilon/2)^2=1-\varepsilon+\varepsilon^2/4>1-\varepsilon,
 \qquad N_0=N^{1-\varepsilon/2}(1+o(1)).
\]
Consequently $N^{1-\varepsilon}\le N_0^{1-\varepsilon/2}\le (N_0-t)^{1-\delta}$ for large $N$. On the expansion event for $\mathbf{G}$, every nonempty set $S$ of at most $N_0^{1-\varepsilon/2}$ rows in one block has at least $(1-\delta)k|S|\ge(1-\varepsilon)(k+t)|S|$ neighbors. Here we used $t\le\varepsilon k/2$. For any set of at most $N^{1-\varepsilon}$ rows across all blocks, apply this bound separately to its intersection with each block. The neighbor sets are disjoint, so their sizes add. The probability that the one sampled $\mathbf{G}$ fails to expand is at most $(N_0-t)^{-1/\delta}\le N^{-1/\varepsilon}$ for large $N$.

By~\eqref{eq:code-stretch} we have $m_0=(N_0-t)^{2d^2+o(1)}$, where $1<2d^2\le\delta^2 k/2<\varepsilon\ell$. In particular $Bm_0>N$ for large $N$, while $Bm_0<B N_0^{\varepsilon\ell}\le N^{\varepsilon\ell}$. In fact $B N_0^{2d^2}\le N^{2d^2}$, so $\log_N m\le 2d^2+o(1)\le\delta^2\ell+o(1)<\ell/36$, since $\delta<1/8$. Sampling one block and making these copies takes polynomial time. For arbitrary fixed real $\varepsilon$, one may use a fixed rational $\delta$ in $(\varepsilon/8,\varepsilon/4]$ and take $B=\lfloor N^\rho\rfloor$ for a fixed rational $\varepsilon/2\le\rho<\varepsilon/(2-\varepsilon)$. Then $(1-\rho)(1-\varepsilon/2)>1-\varepsilon$, so the same cutoff, stretch, and failure bounds hold.

Fix any sampled $\mathbf{G}$ and consider the output in one block on uniform input. With probability $2^{-t}$ all its shared bits are zero. Conditional on this event, write the output as $\mathbf{w}+\mathbf{e}$, where $\mathbf{w}\in\operatorname{im}_{\F_2}\mathbf{G}$ is the parity vector and $\mathbf{e}$ records the majorities. The remaining $r$ inputs of each majority are distinct and uniform. Thus, for every row $i$, Hoeffding's inequality~\cite{hoeffding1963probability} gives
\[
 \Pr[e_i=1\mid\text{shared bits zero}]
 =\Pr[\operatorname{Bin}(r,1/2)>(r+t)/2]
 \le \exp(-t^2/(2r))\le e^{-4}.
\]
The conditional expected weight of $\mathbf{e}$ is at most $e^{-4}m_0$. Markov's inequality gives conditional probability at least $1-e^{-4}/(1/2-\delta)\ge1/2$ that $|{\supp(\mathbf{e})}|\le(1/2-\delta)m_0$. The rank test therefore accepts a generator block with probability at least $\alpha=2^{-t-1}$. On a uniform block it accepts with probability at most $\exp(-\Omega_\varepsilon(m_0))$ by Theorem~\ref{thm:codes}.

Run the rank test on each of the $B$ output blocks and accept if at least $\lceil\alpha B/2\rceil$ tests accept. Under the generator, the restrictions of the seed to distinct blocks are independent even after fixing $\mathbf{G}$. The Chernoff bound puts the probability of rejection at $\exp(-\Omega_\varepsilon(B))$. Under uniform output, the blocks are again independent. A union bound over subsets of accepting blocks gives acceptance probability at most
\[
 2^B\bigl[\exp(-\Omega_\varepsilon(m_0))\bigr]^{\lceil\alpha B/2\rceil}
 \le \exp(-\Omega_\varepsilon(m)).
\]
These estimates hold for every sampled graph, using the specified block and row order.
\end{proof}

\section{Public Key Encryption}\label{sec:pke}

The code from Section~\ref{sec:codes} gives an efficient test for noisy codewords when their coordinates are in the right order. We hide the positions and order of these coordinates by placing the rows of its generator matrix among the edges of a larger random hypergraph. The public key describes the larger hypergraph. The secret key records the planted copy. We first give the noisy-$k$XOR scheme, whose encryption uses the public incidence matrix. Section~\ref{sec:goldreich-pke} gives the Goldreich variant, whose encryption instead evaluates the public XOR-MAJ generator. Both schemes restrict the ciphertext to the hidden coordinates. For noisy $k$XOR the receiver computes one rank. For Goldreich it counts how many hidden blocks pass the rank test.

A \emph{public key encryption scheme} consists of three polynomial-time algorithms. Key generation $\KeyGen(1^n)$ samples a public key $pk$ and a secret key $sk$. For each message $u$, anyone can sample a ciphertext $\mathbf{c}\leftarrow\Enc_{pk}(u)$, and $\Dec_{sk}(\mathbf{c})$ must recover $u$ except with negligible probability over key generation and encryption. For random inputs $X$ and $Y$, the \emph{distinguishing advantage} of a test $D$ is $\bigl|\Pr[D(X)=1]-\Pr[D(Y)=1]\bigr|$, where the probabilities include any randomness used by $D$. We use \emph{chosen-plaintext security}, meaning that every efficient adversary has negligible advantage in distinguishing encryptions of two equal-length messages it chooses after seeing $pk$~\cite{goldwasser1984probabilistic}.

\subsection{The Two Conjectures}

\paragraph{Noisy Parities on a Random Matrix.}
The first conjecture concerns the equations available to an adversary who sees a random public matrix. A single constant specifies both the number of equations and the noise rate. In each conjecture we ask for some constant $c>0$ such that every nonuniform Boolean circuit of size at most $\exp(n^c)$ has distinguishing advantage at most $1/4$. The exponent may depend on the fixed parameters of the conjecture. Polynomial-size circuits suffice for the corresponding ordinary security claim. We take rational witnesses for the existential constants that specify sampling distributions.

\begin{conjecture}[Constant-Noise Random $k$XOR]\label{conj:xor}
There is a constant $0<\delta<1/2$ such that for every sufficiently large fixed arity $k$ there is a $c>0$ with the following property. Take $M=\lfloor n^{\delta k}\rfloor$ independent uniform weight-$k$ rows for $\mathbf{A}\in\F_2^{M\times n}$. Independently sample $\mathbf{x}\sim U_n$ and $\mathbf{e}\sim\Ber(1/2-\delta)^M$. For all sufficiently large $n$, circuits of size at most $\exp(n^c)$ have distinguishing advantage at most $1/4$ between
\[
 (\mathbf{A},\mathbf{A}\mathbf{x}+\mathbf{e})\qquad\text{and}\qquad (\mathbf{A},\mathbf{u}_M),
\]
where the uniform word is independent of $\mathbf{A}$.
\end{conjecture}

For ordinary LPN with uniformly random dense equations, the algorithms of Blum, Kalai, and Wasserman~\cite{blum2003noise} and Lyubashevsky~\cite{lyubashevsky2005parity} give subexponential attacks with different sample requirements. The rows here have constant weight. Search-to-decision reductions for sparse noisy equations also depend on the parameters. Bogdanov, Sabin, and Vasudevan trade sample complexity against success probability at constant noise~\cite{bogdanov2019xor}, while Bogdanov, Rosen, and Tan improve sample efficiency at inverse-polynomial noise, with a change in arity~\cite{bogdanov2025sample}.

The density scale $n^{k/2}$ comes from strong refutation, which certifies that every assignment violates nearly half the equations. Allen, O'Donnell, and Witmer give a polynomial-time algorithm at this scale up to logarithmic factors~\cite{allen2015refute}. With $n^s$ random equations, $1<s<k/2$, the tradeoff of Raghavendra, Rao, and Schramm gives time $\exp(\widetilde O(n^u))$ at $u=(k-2s)/(k-2)$, where $\widetilde O$ hides powers of $\log n$~\cite{raghavendra2017strongly}. The sum-of-squares lower bounds of Kothari, Mori, O'Donnell, and Witmer match this refutation scale up to lower-order factors~\cite{kothari2017sum}. At $s=\delta k$ the exponent $u$ is positive, leaving room for a positive $c<u$ in our decision conjecture.

Further evidence comes from statistical-query lower bounds for planted constraint satisfaction problems~\cite{feldman2015complexity} and time--space tradeoffs for streaming distinguishers of local-generator outputs~\cite{garg2020time}. A statistical-query algorithm accesses approximate expectations of bounded functions; a streaming algorithm processes the input sequentially using limited memory. Conversely, worst-case subexponential attacks on constant-locality generators are known~\cite{unal2023worst}. The usual limited-independence and parity-test bounds follow from essentially the expansion analysis of Applebaum, Barak, and Wigderson~\cite[Theorem~9.1 and Corollary~9.3]{applebaum10}.

\paragraph{Planting a Known Expanding Hypergraph.}
A $k$-uniform hypergraph on $N$ vertices is a set of $k$-element subsets of $[N]$. Write $\mathcal G_k(N,p)$ for the distribution that includes each such subset independently with probability $p$. For a hypergraph $G$, let $\mathbf{A}_G$ be its binary incidence matrix, with one row per edge and one column per vertex. We order its edges lexicographically. For a fixed $k$-uniform $H$ on $n$ vertices, draw $G_0\sim\mathcal G_k(N,p)$ and an independent uniform injection $\pi:[n]\to[N]$. Planting $H$ means adjoining the image of each of its edges to $G_0$, counting an edge already present only once. The distinguisher knows $H$ and sees either $G_0\cup\pi(H)$ or $G_0$. The injection is secret.

\begin{conjecture}[Arbitrary Expanding Hypergraph Planting]\label{conj:plant}
There are constants $C_1>2$, $C_2\ge2$, and $0<\varepsilon_0<1/(2C_2)$ for which the following holds. For every fixed $0<\varepsilon\le\varepsilon_0$ and every sufficiently large fixed $k$, set
\begin{equation}\label{eq:plant-parameters}
 N=\lceil n^{C_1}\rceil,\qquad p=N^{C_2\varepsilon k-k}.
\end{equation}
Fix a constant in the edge budget $O(n^{\varepsilon k})$. There is a $c>0$ such that, for all sufficiently large $n$ and every $(1-\varepsilon,k,n^{1-\varepsilon})$-expanding $H$ on $n$ vertices within this budget, circuits of size at most $\exp(n^c)$ given $H$ distinguish $G_0\cup\pi(H)$ from $G_0$ with advantage at most $1/4$.
\end{conjecture}

The same conjecture then holds with any larger fixed $C_2$ upon decreasing $\varepsilon_0$. Indeed, one can raise the sampling rate from $p$ to $p'>p$ in either experiment by adding each absent edge independently with probability $(p'-p)/(1-p)$. A test after this addition would give a test at the original density. The quantifier over $H$ makes the conjecture independent of the method for constructing the hidden code~\cite{goldwasser2015cryptographic}. To see the connection to log-density~\cite{BCCFV10,dhawan2025detection} for a particular $H$, consider $j$ fixed edges spanning $v$ vertices. They have $(N)_v p^j$ expected labeled embeddings in the background, where $(N)_v=N(N-1)\cdots(N-v+1)$. Expansion gives $v\ge(1-\varepsilon)kj$, so the expectation is at least $N^{(C_2-1)\varepsilon kj+o(1)}$. Section~\ref{sec:lowdegree} bounds overlaps between copies to treat statistics involving many such patterns.

\paragraph{Choosing the Planting Density.}
Take $\delta$ from Conjecture~\ref{conj:xor} and $C_1,C_2,\varepsilon_0$ from Conjecture~\ref{conj:plant}. Choose a sufficiently small fixed rational $0<\varepsilon<\min\{\varepsilon_0,\delta/C_2\}$, in terms of $C_1,C_2$, for the bounds in Sections~\ref{sec:schur}--\ref{sec:lowdegree}. For each sufficiently large fixed $k$, the random background then has $(1+o(1))N^{C_2\varepsilon k}/k!$ edges in expectation, while the noisy-$k$XOR conjecture supplies $\lfloor N^{\delta k}\rfloor$ equations. The first quantity is $o(N^{\delta k})$.

To choose the hidden code, use Theorem~\ref{thm:codes} at a fixed accuracy $\tau<\min\{\varepsilon,\delta\}$ small enough that its arity $k$ is large enough for both conjectures and the lower bounds in Sections~\ref{sec:schur}--\ref{sec:lowdegree}, with $C_2\varepsilon k>4$. Such a choice exists by~\eqref{eq:code-parameters}. Expansion at accuracy $\tau$ implies expansion at accuracy $\varepsilon$. Also $m=n^{2d^2+o(1)}$ and $2d^2\le\tau^2k/2<\varepsilon k$, so $H$ meets the edge budget. The noise rate $1/2-\delta$ lies below the test's guaranteed error fraction $1/2-\tau$.

\subsection{The Scheme and Its Correctness}

Let $\mathsf{Sample}(1^n)$ denote the matrix sampler from Theorem~\ref{thm:codes} at accuracy $\tau$, and let $D$ denote its noisy-codeword test, where output one means that the input is judged to be a noisy codeword. We first analyze a single encryption of one bit, then amplify.

\paragraph{Key Generation.}
Sample $\mathbf{B}\leftarrow\mathsf{Sample}(1^n)$, and regard the support of each row as an edge of a hypergraph $H$ on $[n]$. If two rows have the same support, output a failure symbol. Otherwise, sample $G_0\sim\mathcal G_k(N,p)$ and a uniform injection $\pi:[n]\to[N]$, and set $G=G_0\cup\pi(H)$. The public key is $G$. The secret key records $\mathbf{B}$ and, for each row $i$ of $\mathbf{B}$, the position of its image edge in the ordered edge list of $G$.

\paragraph{Encryption.}
To encrypt $b=0$, sample independent $\mathbf{x}\sim U_N$ and $\mathbf{e}\sim\operatorname{Ber}(1/2-\delta)^{|G|}$, and output $\mathbf{c}=\mathbf{A}_G\mathbf{x}+\mathbf{e}$. To encrypt $b=1$, output $\mathbf{c}\sim U_{|G|}$.

\paragraph{Decryption.}
Use the secret key to extract from $\mathbf{c}$ the $m$ labels on the planted edges, in the original row order of $\mathbf{B}$. Call this vector $\mathbf{y}$. Output $1-D(\mathbf{y})$. On the failure symbol, encryption outputs a fixed string and decryption outputs zero.
\newline
\vspace{3px}

The duplicate-row check ensures that different coordinates of the hidden code occupy different ciphertext positions. Consequently, an encryption of zero gives independent errors on the hidden coordinates, so their total weight concentrates, and an encryption of one restricts to a uniform vector. The test itself accepts arbitrary errors up to its weight threshold. The check is efficient, and expansion on two rows rules out duplicates because $2(1-\tau)k>k$. Theorem~\ref{thm:codes} therefore bounds the probability of failure by $n^{-1/\tau}$. We do not need to test expansion itself.

The honest algorithms use a fixed polynomial number of unbiased random bits. Generate the injection by choosing among the unused vertices in turn. For a uniform choice from a set of size $L$, sample from the next power-of-two range and reject out-of-range values. After $m$ unsuccessful trials return a fixed valid value. For an ideal Bernoulli probability $p$, using $\lfloor 2^m p\rfloor/2^m$ changes a single trial by at most $2^{-m}$ in statistical distance. For the rational exponents in~\eqref{eq:plant-parameters}, integer powering and binary search compute this numerator and $N$ in polynomial time. The same precision suffices for the encryption noise. A union bound over the polynomially many samples gives total statistical error $\exp(-\Omega(m))$. We analyze the ideal distributions below.

\begin{lemma}\label{lem:pke-correct}
For each message bit, the decryption error of this scheme is at most
\[
 n^{-1/\tau}+\exp(-\Omega(m)).
\]
\end{lemma}
\begin{proof}
Condition on any sampled matrix with distinct rows and any planted injection. For an encryption of zero, restricting $\mathbf{x}$ to the coordinates $\pi(1),\ldots,\pi(n)$ gives a uniform vector $\mathbf{x}'\in\F_2^n$. The extracted labels are exactly
\[
 \mathbf{y}=\mathbf{B}\mathbf{x}'+\mathbf{e}',\qquad \mathbf{e}'\sim\operatorname{Ber}(1/2-\delta)^m.
\]
Distinct planted edges occupy distinct ciphertext positions, so the entries of $\mathbf{e}'$ are independent. Since $\tau<\delta$, the Chernoff bound gives $\Pr[|\supp(\mathbf{e}')|>(1/2-\tau)m]\le\exp(-\Omega(m))$. Below this weight Theorem~\ref{thm:codes} guarantees acceptance. For an encryption of one, the same distinctness gives $\mathbf{y}\sim U_m$, which Theorem~\ref{thm:codes} rejects except with probability $\exp(-\Omega_\tau(m))$. Both guarantees apply to every sampled matrix, without conditioning on expansion. Adding the probability of the failure symbol and the statistical error of the sampling implementation proves the claim.
\end{proof}

\subsection{Security}

\begin{theorem}\label{thm:pke}
Assuming Conjectures~\ref{conj:xor} and~\ref{conj:plant}, there is a semantically secure public key encryption scheme. For some constant $\delta'>0$, its decryption error and the chosen-plaintext distinguishing advantage of circuits of size at most $\exp(n^{\delta'})$ are at most $\exp(-n^{\delta'})$. All its algorithms run in polynomial time.
\end{theorem}
\begin{proof}
We compare encryptions of zero and one by first removing the planted code,
then replacing the noisy labels by uniform labels, and finally restoring the
planted code. The first and last steps use the planting conjecture; the middle
step uses the random $k$XOR conjecture. Formally, the four distributions are
\[
 (G,\operatorname{Enc}_G(0)),\quad
 (G_0,\operatorname{Enc}_{G_0}(0)),\quad
 (G_0,\mathbf{u}_{|G_0|}),\quad
 (G,\mathbf{u}_{|G|}),
\]
where $G$ is the planted public key and
$G_0\sim\mathcal G_k(N,p)$. Each experiment also samples $\mathbf{B}$ and performs
the duplicate-row check. If that check fails, it outputs the same failure
public key and fixed ciphertext, so this branch is identical in all four
experiments.

Each of the first and last transitions has distinguishing advantage at most $1/4+o(1)$ by Conjecture~\ref{conj:plant}. The uniformity over $H$ allows us to use the conjecture for our randomly sampled code. Condition on $\mathbf{B}$ satisfying expansion. For every such fixed $\mathbf{B}$, an adversary for the transition gives an adversary for the planting experiment: supply its challenge hypergraph as the public key and generate the ciphertext using public encryption. This requires no knowledge of $\pi$. The bound is uniform over $\mathbf{B}$, so it also holds after averaging over the sampler. The expansion failure probability for $\mathbf{B}$ is at most $n^{-1/\tau}=o(1)$.

The two random instances in the middle differ in how their supports are sampled. The conjecture gives a fixed number of independent rows, whereas the public graph is a set of random size. To pass between them, start with the $M=\lfloor N^{\delta k}\rfloor$ rows and labels from Conjecture~\ref{conj:xor} on $N$ variables. Independently sample $L\sim\operatorname{Bin}(\binom Nk,p)$, the edge count of the random public graph. Our choice of $\varepsilon$ gives $\E L=o(M)$, so Markov's inequality yields $\Pr[L>M]=o(1)$. The probability of a duplicate among the $M$ supports is at most $\binom M2/\binom Nk=O_k(N^{(2\delta-1)k})=o(1)$. Return a fixed output on either exceptional event. Otherwise take the first $L$ rows and sort them and their labels by support. For fixed $L$ and distinct supports, these rows form a uniform $L$-edge subset. In the noisy case they retain the same secret and independent $\Ber(1/2-\delta)$ errors, while in the uniform case their labels remain independent uniform bits. This simulates the two random-graph hybrids up to $o(1)$ statistical distance, with the duplicate-row check on the independently sampled hidden matrix performed in both. The middle transition therefore costs at most $1/4+o(1)$.

Adding the three bounds gives encryption advantage at most $3/4+o(1)$. Choose $0<c<1$ smaller than the planting circuit exponent and smaller than $C_1$ times the noisy-$k$XOR circuit exponent. The latter factor converts a bound on $N$ public variables to one in the $n$ hidden vertices. At this exponent the reductions' polynomial costs fit within both circuit bounds.

Amplified key generation samples $n$ blocks of $n$ independent base key pairs and gives the receiver the corresponding secret keys. To encrypt $b$, in each block choose fresh independent fair bits $r_1,\ldots,r_n$, encrypt each under its corresponding public key with fresh randomness, and send the ciphertexts together with $b\oplus r_1\oplus\cdots\oplus r_n$. Unmasking fails only if some share decrypts incorrectly, with probability at most $n(n^{-1/\tau}+\exp(-\Omega(m)))=o(1)$. Decrypt all blocks and take the majority of the recovered bits. Independence across blocks gives error $\exp(-\Omega(n))$.

For secrecy, the bias $2\Pr[\text{correct}]-1$ for predicting a fair share from its public key and ciphertext is at most $0.8$ at size $S=\exp(n^{c/2})$, by the base distinguishing bound. The quantitative XOR lemma~\cite[Theorem~6.2]{holenstein2006strengthening} bounds the bias for predicting the parity of the $n$ shares by $0.8^n+\eta$ at size $\eta^2 S/\operatorname{poly}(n)$ (see also~\cite[Theorem~6]{holenstein2005one}). To convert a block distinguisher into such a predictor, give it the share ciphertexts, public keys, and a fresh fair mask bit. Its guess for the message, XORed with that mask bit, predicts the parity with bias equal to its distinguishing advantage. With $\eta=\exp(-3n^{\delta'})$ and $0<\delta'<c/2$, the available circuit size still exceeds $\exp(n^{\delta'})$ by a polynomial factor. In a hybrid over the blocks and message bits, the reduction can generate every other encryption using the public keys, even when the messages depend on those keys. The polynomial number of hybrids preserves the claimed advantage.
\end{proof}

The only new hardness conjecture used in this reduction is Conjecture~\ref{conj:plant}. It concerns the position of a known expanding hypergraph inside a larger random one. The random $k$XOR conjecture is used only after removing that hidden structure. The amplification argument preserves a positive subexponential security exponent. Since all lengths are polynomial in $n$, the security bound remains exponential in a positive power of the public-key length.

\subsection{Using Goldreich with XOR-MAJ}\label{sec:goldreich-pke}

The Goldreich scheme uses the expanding instance from Corollary~\ref{cor:xor-maj} as its hidden graph. Encryption evaluates the public XOR-MAJ generator. The roles of the inputs are part of its description. For even $k=2\ell$, a \emph{marked edge} is a pair $(e,L)$ consisting of a $k$-element vertex set and an $\ell$-element subset $L\subset e$. The inputs in $L$ feed the parity gate and the others feed the majority gate. Write $f_G:\F_2^N\to\F_2^{|G|}$ for this generator on a marked graph $G$, ordering the marked edges lexicographically. The ordering of inputs within each half has no effect.

\paragraph{The Marked Planting Conjecture.}
Let $\mathcal G_k^{\mathrm{mark}}(N,p)$ include each possible marked edge independently with probability $p$. Planting again adjoins the image of $H$ under a uniform injection, now sending $(e,L)$ to $(\pi(e),\pi(L))$. Both $H$ and the public graph specify their marks. Expansion counts the distinct vertices in any collection of edges, ignoring the marks.

\begin{conjecture}[Marked Expanding Hypergraph Planting]\label{conj:marked-plant}
For every fixed $0<\varepsilon<1/144$ and every fixed even $k>2/\varepsilon$, there is a constant $c>0$ with the following property. Put $N=n^{10}$ and $p=2^{-\lceil(1-2\varepsilon)k\log_2 N\rceil}$. For all sufficiently large $n$, let $H$ be any $(1-\varepsilon,k,n^{1-\varepsilon})$-expanding marked hypergraph on $n$ vertices with at most $n^{\varepsilon k}$ edges. Given $H$, every nonuniform Boolean circuit of size at most $\exp(n^c)$ has advantage at most $1/4$ in distinguishing $G_0\cup\pi(H)$ from $G_0$, where $G_0\sim\mathcal G_k^{\mathrm{mark}}(N,p)$.
\end{conjecture}

\paragraph{Goldreich on the Random Background.}
Take $k=2\ell$ from Corollary~\ref{cor:xor-maj} at accuracy $\varepsilon$. On $n$ hidden vertices it gives $m<n^{\varepsilon\ell}$ edges, within the budget of Conjecture~\ref{conj:marked-plant}. The expected public edge count is $\Theta_k(N^s)$ with $s=2\varepsilon k=4\varepsilon\ell$. Our choice $\varepsilon<1/144$ gives $s<\ell/36$.

This places the random generator inside the concrete challenge of Applebaum and Lovett~\cite[Section~1.3]{applebaum2016algebraic}, who ask for parity arity at least $5s$ and majority arity greater than $36s$ at output length $N^s$. For these arities their lower bounds for algebraic refutation and inversion give running time at least $\exp(\Omega(N^{1/9}))$ with high probability over the random graph~\cite[Theorem~5.5 and Corollary~5.7]{applebaum2016algebraic}.\footnote{The degree parameter in their bound is the least degree of a nonzero multilinear polynomial over $\F_2$ vanishing on the zeros or on the ones of the predicate. Fix the parity inputs so that this polynomial remains nonzero. The resulting predicate is majority or its complement, for which this degree is $(\ell+1)/2$~\cite[Section~1.2]{applebaum2016algebraic}. Their running-time exponent is therefore at least $1-32(s-1)/(\ell-1)>1/9$.} Related evidence for random local generators includes small-bias results~\cite{mossel2006biased,applebaum2016algebraic}, semidefinite-hierarchy lower bounds~\cite{odonnell2014goldreich}, statistical-query lower bounds~\cite{feldman2015complexity}, and streaming lower bounds~\cite{garg2020time}. These results concern particular classes of attacks. The general-circuit assertion below remains a conjecture.

We formulate the random-Goldreich conjecture on the random marked background. Conditional on the number of edges, this background is a uniform set of marked $2\ell$-element subsets. Since $2s<2\ell$, sampling the same number of constraints with replacement changes the distribution by $o(1)$ statistical distance. Randomly ordering the edge list and the inputs within each half gives the usual ordered-tuple convention for Goldreich's generator; the resulting permutation of output coordinates is public. Standard random-hypergraph expansion bounds apply in this range~\cite[Lemma~2.5]{applebaum2016fast}, but the conjecture averages over the full random distribution without conditioning on expansion. In particular, it is different from the conjectures asserting security on every expander discussed in the introduction. Known attacks also depend on the predicate and density~\cite{couteau2018concrete,fu2026grouping,dantona2026common,unal2023worst}; none of the preceding evidence permits dropping these parameter restrictions.

\begin{conjecture}[Random Goldreich with XOR-MAJ]\label{conj:goldreich}
Take $k=2\ell$ from Corollary~\ref{cor:xor-maj} at the fixed accuracy $\varepsilon$, and use $N$ and $p$ from Conjecture~\ref{conj:marked-plant}. There is a constant $c>0$ such that, for all sufficiently large $n$, every nonuniform Boolean circuit of size at most $\exp(n^c)$ has advantage at most $1/4$ in distinguishing
\[
 (G_0,f_{G_0}(\mathbf{u}_N))\qquad\text{from}\qquad(G_0,\mathbf{u}_{|G_0|}),
\]
where $G_0\sim\mathcal G_{2\ell}^{\mathrm{mark}}(N,p)$.
\end{conjecture}

One challenge vector suffices because a single evaluation on the planted graph already contains all the blocks used by the receiver.

\begin{theorem}\label{thm:goldreich-pke}
Assuming Conjectures~\ref{conj:marked-plant} and~\ref{conj:goldreich}, there is a semantically secure public key encryption scheme. For some constant $\delta'>0$, its decryption error and the chosen-plaintext distinguishing advantage of circuits of size at most $\exp(n^{\delta'})$ are at most $\exp(-n^{\delta'})$. All its algorithms run in polynomial time.
\end{theorem}
\begin{proof}
Use the same key generation as in the noisy-$k$XOR scheme, with the marked instance $H$ from Corollary~\ref{cor:xor-maj} and a marked background $G_0$. The secret key records the positions of the planted edges in their original order. To encrypt zero, output $f_G(\mathbf{x})$ for one uniform $\mathbf{x}\in\F_2^N$. To encrypt one, output a fresh uniform word $\mathbf{u}_{|G|}$. Decryption restricts the ciphertext to the hidden coordinates, applies the test from the corollary, and outputs zero if it accepts.

For a key with distinct planted supports, the restricted ciphertexts have the respective laws of $f_H(\mathbf{u}_n)$ and $\mathbf{u}_m$. In the first case the seed restricted to the planted vertices is uniform. In the second case the secret key selects distinct entries of a uniform word. By Corollary~\ref{cor:xor-maj} the errors are at most $\exp(-\Omega_\varepsilon(n^{\varepsilon/2}))$ and $\exp(-\Omega_\varepsilon(m))$, respectively. The duplicate-support check fails with probability at most $n^{-1/\varepsilon}$ by the same corollary. Thus the base scheme has decryption error $o(1/n)$, including the error from capped sampling.

Choose $0<c<1$ smaller than both circuit exponents in Conjectures~\ref{conj:marked-plant} and~\ref{conj:goldreich}. For circuits of size at most $\exp(n^c)$, the polynomial cost of simulating the three hybrids from the proof of Theorem~\ref{thm:pke} fits within both bounds. Removing and restoring the marked copy each cost at most $1/4+o(1)$ by the marked planting conjecture. On the random public graph the middle hybrid replaces one evaluation by a uniform word, at cost at most $1/4$ by Conjecture~\ref{conj:goldreich}. The base encryption advantage is therefore at most $3/4+o(1)$. The same XOR and majority amplification gives the stated decryption error and chosen-plaintext advantage for any fixed $0<\delta'<c/2$.
\end{proof}

\section{Schur Product Lower Bounds}\label{sec:schur}

Sidelnikov and Shestakov recovered generalized Reed--Solomon structure from a generator matrix~\cite{sidelnikov1992insecurity}, and Wieschebrink used square codes to treat subcodes~\cite{wieschebrink2009cryptanalysis}. Square dimensions distinguish high-rate McEliece keys~\cite{FGO13,moraTillichDualSquare}, disguised Reed--Solomon codes~\cite{CGGOT14,bbcrs2015attack}, wild Goppa codes~\cite{COT14}, and subcodes of algebraic-geometry codes~\cite{couvreur2014subcodes} in the respective attack regimes. For Reed--Muller variants, Schur products can distinguish constituent blocks~\cite{rubenkov2026structural}. Earlier attacks on Sidelnikov's Reed--Muller scheme used minimum-weight words~\cite{sidelnikov94,minder2007cryptanalysis,borodin2014effective}.

For $\mathbf{u},\mathbf{v}\in\F_2^m$, their \emph{Schur product} is the column vector $\mathbf{u}*\mathbf{v}$ with $(\mathbf{u}*\mathbf{v})_i=u_i v_i$ for $i\in[m]$. The $r$-th Schur power of a binary linear code $C$, for $r\ge1$, is
\[
 C^{*r}=\operatorname{span}\{\mathbf{c}_1*\cdots*\mathbf{c}_r:\mathbf{c}_1,\ldots,\mathbf{c}_r\in C\}.
\]
The case $r=2$ is the square code~\cite{randriambololona2015products}. We first consider restrictions of the hidden code, then compare the Schur powers of the two public incidence codes in the planting experiment.

\subsection{Consequences of Expansion}

\begin{lemma}\label{lem:schur-local}
Let $\mathbf{A}\in\F_2^{m\times n}$ have row weight $k$, and suppose every nonempty set of at most $t$ rows has more than $k/2$ times as many neighbors as rows. Every set of at most $t$ rows is linearly independent. Consequently, if $C=\operatorname{im}\mathbf{A}$ and $S,T\subseteq[m]$ are disjoint with $|S|+|T|\le t$, then
\[
 \{\mathbf{c}|_S:\mathbf{c}\in C,\ \mathbf{c}|_T=\mathbf{0}\}=\F_2^S.
\]
All Schur powers of this restriction have dimension $|S|$.
\end{lemma}
\begin{proof}
In a nonempty set of rows summing to zero, every neighboring column occurs at least twice, contrary to expansion for sets of size at most $t$. Thus the restriction of $C$ to $S\cup T$ is onto, including on vectors that are zero on $T$.
\end{proof}

So we may shorten on $T$ and retain just $S$, with both sets chosen after seeing the hidden matrix.

\subsection{Schur Powers of the Public Code}

In the unmarked model the possible edge vectors form the set $\Omega=\{\one_e:e\in\binom{[N]}k\}$. In the marked model an edge has $k/2$ parity inputs $L\subset e$, and we take $\mathbf{x}=\binom{\one_L}{\one_{e\setminus L}}\in\F_2^{2N}$ for each allowed pair $(e,L)$. In either model, identify a graph $G$ with its subset of $\Omega$. The matrix $\mathbf{M}$ with rows $\mathbf{x}^{\mathsf T}$ for $\mathbf{x}\in G$ generates $C(G)=\operatorname{im}\mathbf{M}$ by its columns. For the Goldreich public key this is an auxiliary linear code that records both input roles.

Let $G_0$ contain each point of $\Omega$ independently with probability $p$, and write $\mu=p|\Omega|=N^{a+o(1)}$. The arity $k$ is fixed. In our applications $a=C_2\varepsilon k$ in the unmarked model and $a=2\varepsilon k$ in the marked model. The natural comparison for $C(G_0)^{*r}$ is the Schur power formed on \emph{all} possible edges. Constant row weight gives relations absent from an unrestricted random linear code~\cite{cascudo2015squares}. Let
\[
 W_r=\operatorname{span}_{\F_2}\left\{\mathbf{x}\longmapsto\prod_{i\in I}x_i:1\le |I|\le r\right\}
 \quad\text{on }\Omega,\qquad s_r=\dim W_r.
\]
Here $I$ ranges over subsets of the vertex coordinates, including their roles in the marked case. For instance, the product of two vertex indicators in the unmarked model records which edges contain both vertices. Multilinearity of the Schur product and the binary identity $x_i^2=x_i$ give $C(G)^{*r}=W_r|_G$. In particular, $s_r=O_r(N^r)$.

\paragraph{Sampling Edge Evaluations.}
To determine the dimension after puncturing, we need restriction to the sampled edges to stay injective. We give a bound that will also apply to equations and incidence parities. Let $W$ be a $D$-dimensional space of binary functions on $\Omega$, $D\ge1$, invariant under vertex permutations. The evaluations $f\mapsto f(\mathbf{x})$ span $W^*$, so choose $D$ of them forming a basis. If $U\subset W^*$ has codimension $j$, every permuted basis has at least $j$ elements outside $U$. Vertex permutations act transitively on $\Omega$ in both models. Averaging over them, at least $j|\Omega|/D$ points give evaluations outside $U$. The case $j=1$ is the usual support bound for a transitive function space~\cite[Theorem~B]{goldstein2005inequalities}.

Allow a fixed set $F$ of at most $|\Omega|/(2D)$ points to be excluded, and include each remaining point in $X$ independently with probability $q$. Conditional on its size, expose $X$ in uniform order. At codimension $j>0$, at least $j(|\Omega|/D-|F|)$ of the $|\Omega|-|F|$ eligible points increase the rank. All previously exposed points lie in the current span. Thus the expected codimension contracts at each draw by a factor at most $1-(|\Omega|/D-|F|)/(|\Omega|-|F|)$. Averaging over the binomial size of $X$ gives
\begin{equation}\label{eq:span-sampling}
 \Pr[W\to W|_X\text{ is not injective}]
 \le D\exp\!\bigl(-q(|\Omega|/D-|F|)\bigr).
\end{equation}

Take $F=\varnothing$ and $X=G_0$. Suppose some deletion of at most $t$ points destroys injectivity. Fix one such set after seeing $G_0$, and independently keep each point of $G_0$ with probability $1/2$. With probability at least $2^{-t}$ this new sample misses the entire deletion set and hence also fails injectivity. Its marginal law is independent sampling at rate $p/2$. Applying~\eqref{eq:span-sampling} and dividing by $2^{-t}$ gives
\begin{equation}\label{eq:robust-sampling}
 \Pr[\exists T\subseteq G_0,\ |T|\le t:\ W\to W|_{G_0\setminus T}\text{ is not injective}]
 \le 2^tD\exp(-\mu/(2D)).
\end{equation}

\begin{theorem}\label{thm:schur-public}
Fix a positive integer $r<a$. With probability $1-\exp(-N^{\Omega(1)})$ over $G_0$, every graph $G$ obtained by adding arbitrary edges and deleting at most $\lfloor\mu/(2s_r)\rfloor$ edges satisfies $\dim C(G)^{*r}=s_r$. This holds in both the marked and unmarked models.
\end{theorem}
\begin{proof}
Apply~\eqref{eq:robust-sampling} with $W=W_r$ and $t=\lfloor\mu/(2s_r)\rfloor$. The bound is at most $s_r\exp(-(1-\ln 2)\mu/(2s_r))$, which is $\exp(-N^{\Omega(1)})$ for $r<a$. Adding any points of $\Omega$ can only supply further evaluations.
\end{proof}

\paragraph{Products on a Subcode.}
For any input subspace $U$, all products of $r$ linear forms $\mathbf{x}\mapsto \mathbf{x}^{\mathsf T}\mathbf{v}$, $\mathbf{v}\in U$, lie in $W_r$. Thus Theorem~\ref{thm:schur-public} gives the same product dimension on the sample and on $\Omega$ for every $U$ simultaneously. Shortening on $T$ amounts to taking $U=\{\mathbf{v}:\mathbf{x}^{\mathsf T}\mathbf{v}=0\text{ for }\mathbf{x}\in T\}$. If $T$ depends on the sample, the comparison uses that same $U$. For attacks exploiting products on shortenings and subfield or subspace subcodes, see~\cite{couvreur2025twisted,lavauzelleRennerTwisted,couvreur2022subspace}.

If we forget the marks and merge equal supports, different supports remain independent and each is present with probability $1-(1-p)^{\binom{k}{k/2}}$. Since $k$ is fixed, this projection preserves the exponent $a$, so the low-degree conclusion of Theorem~\ref{thm:schur-public} applies to the projected code as well.

\paragraph{The Remaining Schur Powers.}
Above degree $a$, the number of rows can be smaller than $s_r$. A dimension test at the first such power could detect even a single linear dependence among their evaluations.

\begin{theorem}\label{thm:all-powers}
Write $D=\lfloor n^{1-\varepsilon}\rfloor$ and $r_0=\lceil a\rceil$. Fix any $H$ satisfying the expansion and edge bounds of Conjecture~\ref{conj:marked-plant} or~\ref{conj:plant}, with the corresponding sampling parameters. In the unmarked case take $\varepsilon$ sufficiently small in terms of $C_1,C_2$ and then $k$ sufficiently large. With probability $1-N^{-\Omega(1)}$, both $G=G_0$ and $G=G_0\cup\pi(H)$ satisfy
\[
        \dim C(G)^{*r}=\min\{s_r,|G|\}\qquad\text{for every }r\ge1.
\]
If $C_T$ is obtained by shortening $C(G)$ on any set $T$ of at most $D/4$ coordinates and puncturing any further coordinates, then $C_T^{*r}$ is the whole space on the surviving coordinates for every $r\ge r_0$.
\end{theorem}
\begin{proof}
Set $r=r_0$ and condition on the image $U$ of the $n$ plant vertices. We construct products of $r$ linear forms vanishing on $T$ whose evaluations form a triangular matrix on the remaining rows. Each product will be one at the next row we remove and zero at every other row still present. It can take arbitrary values on rows already removed. Taking $r-1$ of the factors to be vertex indicators keeps just the rows containing those vertices. We need one more linear form to select a single row among them.

First, write $m$ for the number of plant edges. Put $s=\lfloor\varepsilon k\rfloor+1$ and $w=\lfloor(s-\varepsilon k)D/2\rfloor$. Fewer than $D$ plant edges can meet a given $w$-set in at least $s$ vertices, since otherwise their union has size at most $w+(k-s)D<(1-\varepsilon)kD$. A uniform $w$-set contains $s$ specified vertices of an edge with probability $\Theta_{k,\varepsilon}((D/n)^s)$. Averaging gives
\begin{equation}\label{eq:expand-edge-budget}
             m=O_{k,\varepsilon}(D(n/D)^s)=O_{k,\varepsilon}(n^{1+\varepsilon^2 k}).
\end{equation}
Hence $m=o(D^r)$ for small enough $\varepsilon$ and large enough $k$ in the unmarked model. In the marked model $m\le n^{\varepsilon k}$ and $r\ge2\varepsilon k$ give this directly.

Second, except with probability $N^{-\Omega(1)}$, any $j\le D$ background rows use at least $7kj/8$ distinct vertices outside $U$. In exposing $j$ possible edges, each vertex slot hits $U$ or an already seen vertex with probability at most $(n+kj)/(N-k+1)$. Counting lists of distinct edges, choosing the slots where this happens and multiplying by $p^j/j!$, the expected number of exceptions of size $j$ is at most
\[
       \left[\frac{O_k(\mu)}{j}\left(\frac{n+kj}{N-k+1}\right)^{k/8}\right]^j.
\]
This is summable to $N^{-\Omega(1)}$ since $kD=o(n)$ and $a<(1-\log_N n)k/8$ with fixed slack. Combining this estimate with expansion inside $U$, every set of up to $D$ public rows has at least $7k/8$ distinct neighbors per row. In particular each nonempty such set has a vertex occurring just once.

Third, join two background rows when they share at least $r$ vertices. Every connected component has $O_k(\log N)$ rows except with probability $N^{-\Omega(1)}$. For the usual tree count~\cite{cooley2018components}, a possible row has at most $\Delta=A_k\binom kr\binom{N-r}{k-r}$ neighbors, where $A_k=1$ without marks and $A_k=\binom{k}{k/2}$ with marks. The expected number of embedded rooted ordered trees on $j$ distinct sampled rows is at most $\mu(4p\Delta)^{j-1}$. If $r>a$, then $p\Delta=o(1)$. At $r=a$, the sampling rate gives $4p\Delta\le(1+o(1))4A_k\binom kr/(k-r)!<1/2$. Indeed $r\le k/8$, $k>288$, and the largest $\lfloor k/2\rfloor$ factors in $(k-r)!$ each exceed $100$. Take $j$ a sufficiently large multiple of $\log N$.

Write $Z$ for the vertices of $T$, so $|Z|\le kD/4$. An indicator of a vertex outside $Z$ vanishes on $T$ and belongs to the code, also with marks by summing its two role indicators. We remove rows of $G\setminus T$ successively. As long as a remaining row contains $r-1$ vertices outside $Z$, some such $(r-1)$-set $I$ occurs in fewer than $\lfloor D/3\rfloor$ remaining plant rows. Otherwise each occurring $I$ extends on plant rows by at least $7k\lfloor D/3\rfloor/8-|Z|-(r-1)\ge kD/30$ different vertices outside $Z\cup I$. A subset of an occurring $I$ can use any of the same extensions, since the new vertex lies with all of $I$ on a remaining plant edge. We can thus grow an ordered list of $r$ distinct vertices outside $Z$, with at least $kD/30$ choices at each step. Each resulting set on a plant edge has at most $r!$ orders. This gives at least $(kD/30)^r/r!$ distinct $r$-sets outside $Z$, contradicting $m\binom kr=o(D^r)$.

The remaining rows through $I$, together with $T$, are linearly independent. Consider any subcollection. Group its background rows through $I$ outside $T$ by their connected components. Each group has $j=O_k(\log N)\le D$ rows. For $j\ge2$ the second estimate gives $7kj/8$ distinct vertices outside $U$ among at most $kj$ incidences, so at least $2(7kj/8)-kj=3kj/4$ occur just once. None of these lies in $I$, which is shared by all $j$ rows. For a singleton at least $7k/8-(r-1)\ge3k/4$ vertices lie outside $U\cup I$. Such vertices cannot occur in another group, since that would give an overlap of $r$ vertices. If none occurs just once in the subcollection, they must all be covered by $T$, because plant rows lie in $U$. The rows of $T$ cover at most $k|T|$ vertices, leaving at most $4|T|/3$ grouped rows. Including the plant rows and $T$ gives at most $7|T|/3+D/3\le11D/12$ rows, below the cutoff $D$ at which we already proved a unique neighbor for every nonempty subcollection. This proves independence. A linear form can therefore select one row through $I$ and be zero on $T$ and the other remaining rows through $I$. Multiplying by the $r-1$ indicators of $I$ isolates this row among all remaining rows.

Finally process rows with fewer than $r-1$ vertices outside $Z$, in decreasing order of the size of their vertex set $J$ outside $Z$. Rows with exactly this set $J$ outside $Z$, together with $T$, use at most $kD/4+r-2<7kD/8$ vertices. If there were $D$ such rows (including $T$), any $D$ of them would need $7kD/8$ neighbors. Thus there are fewer than $D$, and they are independent. Use a linear selector as above and multiply by the indicators of $J$, repeating a factor to reach degree $r$. Any additional nonzero entries belong to rows already treated. The products have triangular evaluations, hence span the whole remaining space. Puncturing restricts this space. Over $\F_2$, repeating any factor preserves its product, so the Schur powers are nested. Theorem~\ref{thm:schur-public} supplies the lower powers.
\end{proof}

Thus the joint list of all power dimensions records only the edge count, up to the stated failure probability. The two edge counts have statistical distance $O(m/\sqrt\mu)$ by the binomial calculation in Section~\ref{sec:dual-hull}. This tends to zero by~\eqref{eq:expand-edge-budget} in the unmarked range and by the marked edge budget otherwise.

\paragraph{Conductors and Closures.}
For two codes $U,V$ on the same positions, their \emph{conductor} is the space of multipliers $\mathbf{z}$ such that $\mathbf{z}*\mathbf{u}\in V$ for every $\mathbf{u}\in U$. Taking $U=C^{*(t-1)}$ and $V=C^{*t}$ gives the $t$-closure of $C$, used by Couvreur, M\'arquez-Corbella, and Pellikaan to recognize subcodes of algebraic-geometry codes~\cite[Definition~1]{couvreur2014subcodes}.

Here an entire conductor can be determined on the complete edge set. Suppose $U,V$ consist of polynomial functions on $\Omega$ of degrees at most $i,j$, respectively, where $i,j\ge1$ and $i+j<a$. Assume that for each $\mathbf{x}\in\Omega$ some $u\in U$ is nonzero at $\mathbf{x}$. With high probability restriction identifies $\operatorname{Cond}(U,V)$ with $\operatorname{Cond}(U|_G,V|_G)$, simultaneously for all such spaces. Indeed, by~\eqref{eq:robust-sampling} the sample determines all polynomial functions of degree at most $i+j$, including constants. For a multiplier $\mathbf{z}$ on $G$ each $\mathbf{z}*(u|_G)$ has a unique lift $v_u\in V$. The identities $u v_{u'}=u'v_u$ hold on $G$ and have degree at most $i+j$, so they hold on $\Omega$. At every $\mathbf{x}$ we can therefore define the lifted multiplier as $v_u(\mathbf{x})/u(\mathbf{x})$ using any $u$ nonzero there. It sends each $u'$ to $v_{u'}$, and its uniqueness follows from the same formula. Every $W_{t-1}$ has full support for $t\ge2$. Thus the $t$-closure of the public code has its complete-layer dimension for every $t\ge2$ with $2t-1<a$. For graph-dependent input spaces the comparison uses the same $U,V$ on both sets.

\section{Interpolation and the Dual of the Public Code}\label{sec:interpolation}

We next study polynomial equations on generator rows and the order to which they vanish. In both experiments, enough random rows exclude a high-order Hermite interpolant in the degree ranges below when the generator columns form a basis of the incidence code. On the dual, disjoint spanning sets of incidence rows instead let us construct an interpolant at a prescribed position.

\subsection{Equations Satisfied by the Public Edge Vectors}\label{sec:edge-equations}

\begin{lemma}\label{lem:determining-edges}
Fix a positive integer $r<a$. With probability $1-\exp(-N^{\Omega(1)})$ over $G_0$, the following holds simultaneously after adding arbitrary points of $\Omega$ and deleting at most $\lfloor\mu/(2(s_r+1))\rfloor$ points. Every binary polynomial of degree at most $r$ vanishing on the remaining points vanishes on all of $\Omega$.
\end{lemma}
\begin{proof}
Apply~\eqref{eq:robust-sampling} to $W_r+\langle1\rangle$, which has dimension at most $s_r+1=O_r(N^r)$. The failure bound is at most $(s_r+1)\exp(-(1-\ln2)\mu/(2(s_r+1)))$.
\end{proof}

In fixed vertex coordinates, the whole space of quadratic equations is therefore determined by $\Omega$ with high probability in either experiment. This fixes the ranks of its forms and, for a fixed basis $Q_1,\ldots,Q_m$, the space of tuples of linear forms satisfying $\sum_i L_iQ_i=0$ as a polynomial. These relations are called \emph{syzygies}. Quadratic-form and syzygy distinguishers for Goppa codes examine the dual or a shortening of it~\cite{couvreurMoraTillichQuadrics,randriambololonaSyzygy}. Here the equations are on the primal incidence rows, and Corollary~\ref{cor:incidence-betti} identifies their successive linear relations.

\subsection{Hermite Interpolation on the Incidence Rows}\label{sec:hermite}

Ghoshal, Ishai, Jain, and Sun~\cite[Sections~2.4 and~4]{GIJS26} use a polynomial that vanishes to high order at the rows outside a set of held-out positions. For a polynomial $f$ over $\F_2$, vanishing to order $s$ at $\mathbf{y}$ means that $f(\mathbf{y}+\mathbf{Z})$ has zero coefficient on every monomial in $\mathbf{Z}$ of degree less than $s$. These coefficients, viewed as functions of $\mathbf{y}$, are its \emph{Hasse derivatives}. They depend on the formal polynomial, so reducing $Y_i^2$ to $Y_i$ can change the order of vanishing. For a multilinear polynomial this amounts to vanishing on the binary Hamming ball of radius $s-1$ about $\mathbf{y}$~\cite[Section~2.4]{GIJS26}. Their asymptotic binary distinguisher uses homogeneous multilinear polynomials of degree $s+1$, with $s$ logarithmic in the code length. The degree grows, but the derivatives of order $s-1$ are quadratic. By Lemma~\ref{lem:determining-edges}, if they vanish on the retained incidence rows, they vanish on every allowed edge. We use this to rule out even homogeneous forms with repeated variables.

A linear dependence among the columns of a generator matrix gives a nonzero linear form $L$ that vanishes at every row. Its power $L^d$ then vanishes to order $d$ at every row. We therefore take the columns of the generator matrix to be a basis of the code under consideration.

\begin{theorem}\label{thm:hermite-rows}
Suppose $a>2$. In the marked model assume $k/2\ge2$, and in the unmarked model assume $k\ge2$. With probability $1-\exp(-N^{\Omega(1)})$ over $G_0$ the following holds simultaneously. Add arbitrary edges, withhold at most $\lfloor\mu/(2(s_2+1))\rfloor$ edges, and take any subcode of the incidence code on the retained edges. For every matrix whose columns form a basis of this subcode, the only homogeneous form of degree $d$ vanishing to order $d-1$ at every row is the zero polynomial, for every $d\ge4$ in the unmarked model and every $d\ge7$ in the marked model.
\end{theorem}
\begin{proof}
Work on the event in Lemma~\ref{lem:determining-edges} for $r=2$, and restrict $\mathbf{M}$ to the retained rows. A generator matrix whose columns form a basis of the subcode can be written $\mathbf{M}\mathbf{B}$, where each column of $\mathbf{B}$ gives the coefficients of one basis word. The row at $\mathbf{x}$ is $\mathbf{x}^{\mathsf T}\mathbf{B}$, so the point at which we evaluate is $\mathbf{B}^{\mathsf T}\mathbf{x}$. Since $\mathbf{M}\mathbf{B}$ has independent columns, these points span the target of $\mathbf{B}^{\mathsf T}$, and $\mathbf{B}^{\mathsf T}$ has a linear right inverse. Choose a basis from the retained points and use it for coordinates. This invertible change preserves homogeneity and the order of vanishing.

We first show that $f$ is multilinear in these coordinates. Each unit vector is now one of the retained points. For a monomial $\mathbf{y}^{\boldsymbol{\alpha}}$ of degree $d$ with $\alpha_i\ge2$, look at the coefficient of $\prod_{j\ne i} Z_j^{\alpha_j}$ in $f(\mathbf{e}_i+\mathbf{Z})$, where $\mathbf{e}_i$ is the $i$-th unit vector. The powers of all variables other than the $i$th are fixed by this coefficient. Homogeneity then fixes the $i$th power as well, giving exactly $[\mathbf{y}^{\boldsymbol{\alpha}}]f$. The term in $f(\mathbf{e}_i+\mathbf{Z})$ has degree $d-\alpha_i\le d-2$, so its coefficient vanishes. This removes every monomial with a repeated variable.

Put $P(\mathbf{x})=f(\mathbf{B}^{\mathsf T}\mathbf{x})$ in the original vertex variables, expressing $\mathbf{B}^{\mathsf T}$ in the chosen coordinates. Since $P(\mathbf{x}+\mathbf{Z})=f(\mathbf{B}^{\mathsf T}\mathbf{x}+\mathbf{B}^{\mathsf T}\mathbf{Z})$, it vanishes to order $d-1$ at the retained points. Changing coordinates can reintroduce repeated variables. We will first eliminate the squarefree degree-$d$ terms of $P$, then return to the coordinates in which $f$ is multilinear. Fix such a monomial with index set $I$. Choose four of its indices $u,u',v,v'$, all in the same role for marked edges. The stated bounds on $d$ ensure that this is possible. Differentiate $P$ once in each variable of $I\setminus\{u,v\}$ and call the resulting quadratic $g$. This derivative has order $d-2$, so $g$ vanishes on the sample and hence on all of $\Omega$. Its coefficient on $x_ux_v$ is $[x_I]P$, with $x_I=\prod_{i\in I}x_i$. Also, differentiating once in characteristic two leaves an even exponent in that variable. Consequently neither $u'$ nor $v'$ occurs in a squarefree quadratic term of $g$.

Consider four edges, choosing one vertex from $\{u,u'\}$ and one from $\{v,v'\}$, with all other vertices fixed outside these pairs. In the marked case both choices are in the selected role and the other role is fixed. There is room to do this for $N\ge k+2$. This gives a finite difference in the two independent vertex choices. Sum the four evaluations of $g$ over $\F_2$. A monomial missing either pair contributes zero, since its value is repeated twice. A quadratic term meeting both pairs must choose one variable from each, and the only possible term with nonzero coefficient is $x_ux_v$. The sum is zero, proving $[x_I]P=0$.

Thus every monomial left in $P$ is divisible by the square of a variable. This property survives a linear substitution over $\F_2$, because the square of a linear form is a sum of squares. Substituting a right inverse of $\mathbf{B}^{\mathsf T}$ recovers $f$ from $P$, so every monomial of $f$ also has a repeated variable. We have already proved that $f$ is multilinear. Hence $f=0$.
\end{proof}

Theorem~\ref{thm:hermite-rows} also permits arbitrary shortening. Restore the shortening positions as zero rows of a subcode generator. A degree-$d$ homogeneous form already vanishes to order $d$ there. On the common-background event of the theorem, the existence bit is therefore zero for both keys, whether one asks for any nonzero form or for a form taking value one at a designated holdout. Its distinguishing advantage is at most $\exp(-N^{\Omega(1)})$, even for a graph-dependent choice of subcode and allowed holdouts.

Independent random puncturing permits much larger sets of omitted rows. Suppose the plant has $m\ge1$ edges, $\mu/m\ge N^{2+\xi}$ for a fixed $\xi>0$, and the test sees only the graph obtained by retaining edges independently at a graph-independent rate $\rho$. Allow up to $(\log N)^b$ holdouts for any fixed $b\ge0$. If $\rho\mu\le N^2(\log N)^{b+2}$, common retention decisions give identical graphs except with probability at most $\rho m\le N^{-\xi}(\log N)^{b+2}$. Otherwise~\eqref{eq:robust-sampling} applied to quadratics at rate $\rho p$ has failure probability $\exp(-\Omega((\log N)^2))$ even after these holdouts. Thus the Hermite existence bit has advantage at most $N^{-\xi+o(1)}$ at every rate, using a common output convention if too few rows survive. The ratio hypothesis holds for our marked parameters and in the unmarked model whenever $C_2\varepsilon k-\varepsilon k/C_1>2$.

\subsection{The Dual Square and the Code Hull}\label{sec:dual-hull}

The dual public code $C(G)^\perp=\ker \mathbf{M}^{\mathsf T}$ consists of the parity checks on the edge coordinates. Its square and its intersection with the primal code give two different statistics. For the square it suffices to find, at every edge, two parity checks that meet only there.

\begin{proposition}\label{prop:dual-square}
Let $h=\dim\operatorname{span}(\Omega)$. If $a>1$, then with probability $1-\exp(-N^{\Omega(1)})$ over $G_0$, every graph $G$ obtained by adding arbitrary edges and deleting at most $\lfloor\mu/(6h)\rfloor$ edges satisfies $(C(G)^\perp)^{*2}=\F_2^{|G|}$.
\end{proposition}
\begin{proof}
Assign each background edge independently to one of $b$ classes, uniformly. Within each class the points of $\Omega$ are sampled independently at rate $p/b$. Apply~\eqref{eq:robust-sampling} to the $h$-dimensional space of linear forms on $\operatorname{span}(\Omega)$ restricted to $\Omega$. A union bound gives
\begin{equation}\label{eq:spanning-classes}
 \Pr[\text{some class fails to span after at most }u\text{ deletions}]
 \le b\,2^u h\exp(-\mu/(2bh)).
\end{equation}
Take $b=3$ and $u=\lfloor\mu/(6h)\rfloor$. Since $h\le2N$ and $a>1$, the failure probability is $\exp(-N^{\Omega(1)})$. For every remaining edge $e$, at least two spanning classes omit it, including when $e$ was added by planting. Expressing $\mathbf{x}_e$ in each class gives two parity checks intersecting only at $e$. Their product is the unit vector at $e$.
\end{proof}

With this probability the dual-square dimension on the full graph records just the edge count. For a plant of $m=o(|\Omega|)$ edges and $0<p<1$, any test of this count has advantage $O(m\sqrt{(1-p)/\mu})$. Indeed, if $B\sim\operatorname{Bin}(|\Omega|-m,p)$ counts background edges off the plant, the planted count is $B+m$ and the null count is $B+\operatorname{Bin}(m,p)$ with an independent binomial summand. The statistical distance between $B$ and $B+1$ is its largest point mass, by unimodality. This is $O(1/\sqrt{(|\Omega|-m)p(1-p)})$, and the expected number of unit shifts is $m(1-p)$.

The \emph{hull} $C(G)\cap C(G)^\perp$ has been used in algorithms for code equivalence~\cite{sendrierSupportSplitting}. A codeword $\mathbf{M}\mathbf{v}$ belongs to the hull exactly when $\mathbf{M}^{\mathsf T}\mathbf{M}\mathbf{v}=\mathbf{0}$, so the hull dimension is $\rank \mathbf{M}-\rank(\mathbf{M}^{\mathsf T}\mathbf{M})$. An entry of this Gram matrix is the parity of a single-vertex or pair incidence, with roles recorded in the marked model. We compare the law of the entire matrix.

\begin{proposition}\label{prop:public-gram}
Suppose $a>2$, $p\le 1/2$, and the number of planted edges is $o(|\Omega|/N^2)$. Under either the null or the planted law, $\mathbf{M}^{\mathsf T}\mathbf{M}$ is within statistical distance $\exp(-N^{\Omega(1)})$ of uniform on the binary span of $\{\mathbf{x}\mathbf{x}^{\mathsf T}:\mathbf{x}\in\Omega\}$. The distributions of the hull dimension under the two laws are also within $\exp(-N^{\Omega(1)})$ of each other.
\end{proposition}
\begin{proof}
Write the Gram matrix as $\sum_{\mathbf{x}\in G}\mathbf{x}\mathbf{x}^{\mathsf T}$. More generally, let $\mathbf{v}(\mathbf{x})$ be a binary vector whose values span a space $V$ of dimension $h\ge1$, and suppose vertex permutations induce linear automorphisms of $V$ carrying $\mathbf{v}(\mathbf{x})$ to the corresponding permuted value. Condition on the planted image $F$, taking $F=\varnothing$ under the null law. Outside $F$, first sample a set $X$ at rate $2p$, then keep each of its points with an independent fair bit. This gives exactly the background at rate $p$ outside $F$. Conditional on $X$, the sum of the kept vectors is uniform on $\operatorname{span}\{\mathbf{v}(\mathbf{x}):\mathbf{x}\in X\}$. Whenever this span is $V$, adding the forced sum on $F$ preserves uniformity on $V$.

Thus the statistical distance from uniform is at most the probability that $X$ fails to span. Apply~\eqref{eq:span-sampling} at $q=2p$ to the $h$-dimensional invariant space of functions $\mathbf{x}\mapsto\lambda(\mathbf{v}(\mathbf{x}))$, $\lambda\in V^*$. For $|F|\le|\Omega|/(2h)$ this gives
\begin{equation}\label{eq:parity-mixing}
 \operatorname{dist}_{\rm TV}\!\left(\sum_{\mathbf{x}\in G}\mathbf{v}(\mathbf{x}),\ \mathrm{Unif}(V)\right)
 \le h\exp(-\mu/h).
\end{equation}
The estimate is uniform in $F$. Apply it to the vector $\mathbf{v}(\mathbf{x})$ obtained by stacking the columns of $\mathbf{x}\mathbf{x}^{\mathsf T}$. Then $h=O(N^2)$ and the size hypothesis on the plant applies. Since $a>2$, the right side is $\exp(-N^{\Omega(1)})$. By Lemma~\ref{lem:determining-edges} in degree one, $\rank \mathbf{M}=\dim\operatorname{span}(\Omega)$ except with the same failure bound. The hull formula gives the second assertion.
\end{proof}

At our scheme parameters $4<a<k/2$ and $|E(H)|=o(N^{k/2})$, as required. The statistic here is the hull dimension.

\paragraph{Higher-Order Incidence Parities.}
For a fixed integer $r\ge1$, let $\mathbf{v}_r(\mathbf{x})$ collect the squarefree monomials of degrees zero through $r$ evaluated at $\mathbf{x}$. The vector $\sum_{\mathbf{x}\in G}\mathbf{v}_r(\mathbf{x})$ lists the parities of the numbers of edges containing each specified set of up to $r$ incidences, including their roles in the marked model. Its constant coordinate gives the parity of $|G|$.

\begin{corollary}\label{cor:incidence-parities}
Suppose $\mu/(N^r\log N)\to\infty$, $p\le1/2$, and $|E(H)|=o(|\Omega|/N^r)$. Under both the planted and null laws, the joint distribution of $\sum_{\mathbf{x}\in G}\mathbf{v}_r(\mathbf{x})$ is within statistical distance $\exp(-\Omega_{k,r}(\mu/N^r))$ of uniform on the binary span of $\{\mathbf{v}_r(\mathbf{x}):\mathbf{x}\in\Omega\}$.
\end{corollary}
\begin{proof}
Vertex permutations permute the monomial coordinates, and their evaluation vectors span a space of dimension $h=O_r(N^r)$. Apply~\eqref{eq:parity-mixing}. The density assumption gives $h\exp(-\mu/h)\le\exp(-\mu/(2h))$ for large $N$.
\end{proof}

At our scheme parameters every fixed $r<a$ meets both the density and plant-size bounds, since $r<a<k/2$. The parities here can have large degree as real polynomials in the edge indicators used in Section~\ref{sec:lowdegree}.

\subsection{Higher Relations among the Equations on the Edge Vectors}\label{sec:incidence-syzygies}

Randriambololona distinguishes a shortening of the dual of a Goppa code by counting the independent linear relations among its quadratic equations, then the relations among those relations, and so on~\cite[Sections~2 and~5]{randriambololonaSyzygy}. On the public incidence code, Lemma~\ref{lem:determining-edges} determines these numbers at every step.

To state this precisely, choose coordinates on the linear span of $\Omega$ and let $P$ be the polynomial ring over $\F_2$ in these coordinates. For $X\subseteq\Omega$, write $R_X$ for the homogeneous coordinate ring of its edge vectors. Its degree-$j$ part $(R_X)_j$ consists of the homogeneous polynomials of degree $j$, with two polynomials identified when they agree on $X$. Take minimal homogeneous generators of the equations on $X$, then of the polynomial relations among those generators, and repeat. The \emph{graded Betti number} $\beta_{i,i+j}(R_X)$ counts the generators of degree $i+j$ at step $i$, counting the equations as step one. When $X$ spans the chosen space, $\beta_{1,2}$ counts the independent quadratic equations, and $\beta_{2,3}$ counts the independent relations among them with linear coefficients. The numbers $\beta_{i,i+1}$ continue this count through all successive linear relations and form the \emph{linear strand}.

\begin{corollary}\label{cor:incidence-betti}
Fix an integer $2\le r<a$. With probability $1-\exp(-N^{\Omega(1)})$ over the background, the following holds after arbitrary edge additions and deletion of at most $\lfloor\mu/(2(s_r+1))\rfloor$ edges. If $X$ is the remaining set of edge vectors, then for every $i\ge1$ and $1\le j\le r-1$,
\[
  \beta_{i,i+j}(R_X)=\beta_{i,i+j}(R_\Omega).
\]
In particular, taking $r=2$ fixes the entire linear strand in both the marked and unmarked models.
\end{corollary}
\begin{proof}
Every linear coordinate on the span of $\Omega$ extends to the ambient vertex coordinates. Thus Lemma~\ref{lem:determining-edges} applies in our chosen coordinates. It shows that $X$ spans the same space as $\Omega$, and that the natural restriction $R_\Omega\to R_X$ identifies the homogeneous pieces through degree $r$, including their multiplication maps in this range.

Let $U=P_1$ be the linear forms. In a basis of $U$, the wedges of $i$ distinct basis elements form a basis of $\bigwedge^i U$. By the Koszul formula~\cite[Proposition~2.7]{eisenbudGeometrySyzygies}, $\beta_{i,i+j}(R_X)$ is the dimension of the kernel of the second map modulo the image of the first in
\[
 \bigwedge^{i+1}U\otimes(R_X)_{j-1}
 \longrightarrow \bigwedge^i U\otimes(R_X)_j
 \longrightarrow \bigwedge^{i-1}U\otimes(R_X)_{j+1}.
\]
In characteristic two the differential removes each wedge factor in turn, multiplies $g\in(R_X)_j$ by that factor, and sums the results. Two successive maps compose to zero since every pair of removals occurs twice. For $j+1\le r$ these multiplications agree on $X$ and $\Omega$.
\end{proof}

\paragraph{Relations on the Dual Rows.}
The distinguisher of~\cite{randriambololonaSyzygy} applies to a shortening of the dual. On $v$ retained positions this is $K=C^\perp$, where $C$ is the punctured primal of rank $h$. An \emph{information set} consists of $h$ independent rows of a generator matrix for $C$. If $b\ge3$ such sets are disjoint, the linear-strand entries of the dual through $\beta_{b-2,b-1}$ are determined by $v,h$. On the full public graph with $a>1$,~\eqref{eq:spanning-classes} gives $h=\dim\operatorname{span}(\Omega)$ and $b=\lfloor\mu/(h\log^2 N)\rfloor$ disjoint sets with high probability.

Let $R$ be the coordinate ring of the dual rows. The evaluations of linear forms on these rows are exactly the linear combinations of the dual generator columns, so evaluation identifies $R_1$ with $K$. Three disjoint information sets give $R_2=\F_2^v$ by Proposition~\ref{prop:dual-square} and its proof. Multiplying coordinate units by linear forms nonzero there gives $R_j=\F_2^v$ for all $j\ge2$. For $j\ge3$ the three Koszul spaces above therefore split by coordinate. At $e$ the differential removes each wedge factor $\mathbf{z}$ with coefficient $z_e$. In a basis of $K$ with just one word nonzero at $e$, it deletes that word from wedges containing it and kills the others. Thus the complex is exact at each coordinate and $\beta_{i,i+j}=0$ for $j\ge3$.

For $j=2$ and $i\le b-3$ the cycles at $e$ are again spanned by wedges of $i$ words of $K$ zero at $e$. Choose one of the $b$ information sets missing $e$ and call it $T$. For each $f\notin T$ express the primal incidence row at $f$ as a sum of rows indexed by $T$, giving a check $\mathbf{z}_f$ supported on $T\cup\{f\}$ with $(\mathbf{z}_f)_f=1$. These checks form a basis of $K$, so the wedges $\mathbf{z}_{f_1}\wedge\cdots\wedge \mathbf{z}_{f_i}$ with $f_j\ne e$ span the cycles at $e$. Choose a second information set among the $b$ disjoint ones missing $T,e,f_1,\ldots,f_i$, and use it to form a check $\mathbf{w}$ equal to one at $e$. The support of $\mathbf{w}$ meets that of $\mathbf{z}_e$ only at $e$ and misses each $\mathbf{z}_{f_j}$. Thus $\mathbf{w}*\mathbf{z}_e=\one_{\{e\}}$ and $\mathbf{w}*\mathbf{z}_{f_j}=\mathbf{0}$. Only the first wedge factor contributes. It follows that the boundary of $(\mathbf{z}_e\wedge \mathbf{z}_{f_1}\wedge\cdots\wedge \mathbf{z}_{f_i})\otimes \mathbf{w}$ is that wedge at $e$ and zero elsewhere. This proves $\beta_{i,i+2}=0$ for $i\le b-3$. We can now read off the remaining entries from the dimensions of the graded pieces. Taking the alternating sum of dimensions in the Koszul complex of total degree $q$ gives $\sum_i(-1)^i\beta_{i,q}$ as the coefficient of $t^q$ in $(1-t)^{v-h}\sum_{j\ge0}(\dim R_j)t^j$~\cite{eisenbudGeometrySyzygies}. In span coordinates there are no linear equations, so generators at step $i>0$ have degree at least $i+1$. For $2\le q\le b-1$ only the linear-strand entry can contribute, and the dimensions $1,v-h,v,v,\ldots$ determine that coefficient.

\subsection{Hermite Interpolation on a Shortened Dual Code}\label{sec:dual-hermite}

Our incidence-row bound excluded Hermite interpolants on the primal code. On a shortened dual the same test can have the opposite outcome in both models. Higher-order vanishing on dual generator rows is itself a structural test~\cite{hemmertWiemersVanishing}. Here the required interpolant comes from $d$ parity checks meeting only at the held-out position.

Keep a set $S$ of public edges, and shorten $C^\perp$ by requiring its words to vanish outside $S$ before deleting those positions. If $\mathbf{M}_S$ consists of the retained rows of $\mathbf{M}$, the shortened dual is $\ker \mathbf{M}_S^{\mathsf T}$. Indeed, a vector in this kernel extends by zeros to a parity check of $C$. Let $\mathbf{K}$ have as columns a basis of this shortened dual, with row $\mathbf{y}_e^{\mathsf T}$ at each retained edge $e$. The test at $e$ asks whether a homogeneous degree-$d$ polynomial takes value one at $\mathbf{y}_e$ and vanishes to order $d-1$ at every $\mathbf{y}_f$ for $f\ne e$. A \emph{systematic basis} is a choice of basis words for which the submatrix on a specified set of rows is the identity.

\begin{lemma}\label{lem:dual-separating-form}
Suppose $\mathbf{x}_e$ can be written in $d$ ways as a sum of other retained edge vectors, using disjoint sets of summands. Then the test at $e$ has answer one. The polynomial can also be taken multilinear in any systematic basis whose identity rows avoid $e$.
\end{lemma}
\begin{proof}
Each expression gives a word of $\ker \mathbf{M}_S^{\mathsf T}$ whose support consists of $e$ and the edges in that expression. Write these words as $\mathbf{K}\mathbf{b}_1,\ldots,\mathbf{K}\mathbf{b}_d$. The polynomial
\[
                 f(\mathbf{Y})=\prod_{j=1}^{d}(\mathbf{b}_j^{\mathsf T}\mathbf{Y})
\]
is one at $\mathbf{y}_e$. At any other row, at most one of these words is nonzero. At least $d-1$ factors vanish there, giving the required order of vanishing. Since $\mathbf{x}_e\ne\mathbf{0}$, the shortened dual contains no nonzero word supported on $\{e\}$. Thus the rows of $\mathbf{K}$ away from $e$ have full rank, since otherwise $\mathbf{K}\mathbf{b}$ would be such a word for some $\mathbf{b}\ne\mathbf{0}$. We can therefore choose identity rows away from $e$. In any such systematic basis, each variable has a nonzero coefficient in at most one factor, by disjointness at the identity rows. The product is then multilinear.
\end{proof}

For $a>1$ and $2\le d=O(\log N)$, use~\eqref{eq:spanning-classes} with $b=d+1$ and $h=\dim\operatorname{span}(\Omega)$. After arbitrary additions and up to $\lfloor\mu/(2(d+1)h)\rfloor$ deletions, every retained edge is omitted by at least $d$ spanning classes except with probability $\exp(-N^{\Omega(1)})$.

The primal minimum distance $\Delta$ allows more deletions. Removing fewer than $\Delta$ positions preserves the row span, since a rank drop would give a nonzero word supported there. If $|T|+1+(d-1)\rank \mathbf{M}<\Delta$, remove $T$ and $e$, express $\mathbf{x}_e$ on a basis, and discard its summands before each subsequent expression. This gives $d$ disjoint expressions. At $a>2$, Lemma~\ref{lem:determining-edges} in degree one gives $\Delta\gg N\log N$, allowing $|T|\le\lfloor\Delta/2\rfloor$ for logarithmic $d$ with high probability.

Independent shortening also permits several held-out positions and growing degree.

\begin{corollary}\label{cor:random-dual-hermite}
Let the plant have $m\ge1$ edges, let $h=\dim\operatorname{span}(\Omega)$, and retain each public edge independently at a rate $\rho$ chosen independently of the graph. Fix a degree bound $d\ge2$, polynomial in $N$. From the retained graph choose any set of positions and prescribe a binary value at each. For every $2\le j\le d$, ask whether a homogeneous degree-$j$ form on the shortened dual generator rows takes those values and vanishes to order $j-1$ at every other row. Uniformly in $\rho$, the joint feasibility answers have distinguishing advantage at most
\[
        O\!\left(\frac{mh(d+1)\log^2 N}{\mu}\right)+\exp(-\Omega(\log^2 N)).
\]
The test uses the same selection rule in both experiments, reports only feasibility, and uses a common answer on the empty graph. For one held-out position the conclusion also holds for multilinear forms in a systematic basis with the identity rows away from that position.
\end{corollary}
\begin{proof}
At a low retention rate we couple the two retained graphs directly. If $\rho\mu\le 4h(d+1)\log^2 N$, use the same background and the same retention coins in the two experiments. The retained graphs agree unless a planted edge survives, with probability at most $\rho m$. At a higher rate the graphs may differ, but we can interpolate every requested assignment on either graph. Indeed,~\eqref{eq:spanning-classes}, at mean $\rho\mu$ with $b=d+1$ and $u=0$, gives $d+1$ spanning classes except with probability $\exp(-\Omega(\log^2 N))$. For each edge $e$ and each $j\le d$ at least $j$ of the classes omit $e$. Lemma~\ref{lem:dual-separating-form} gives a degree-$j$ form equal to one at $e$ and vanishing to order $j-1$ at every other row. Summing these forms for the prescribed ones realizes every assignment. For a single position the lemma gives the stated multilinearity. Thus all feasibility answers agree in the high-rate case, even if the prescribed positions depend on the retained graph.
\end{proof}

In particular, if $\mu/m\ge N^{1+\xi}$ for some fixed $\xi>0$, the advantage is at most $N^{-\xi+o(1)}$ through logarithmic degree and tends to zero even through degree $N^{\xi-\eta}$ for any fixed $0<\eta<\xi$. This includes the marked scheme parameters and, by~\eqref{eq:expand-edge-budget}, the unmarked parameters for sufficiently small $\varepsilon$ and large $k$. For a shortening chosen from the whole graph, one can instead retain a single position from the support of a nonzero primal word and all positions outside its support. Orthogonality forces every surviving dual word to be zero at that position.

\section{Low-Degree Lower Bounds}\label{sec:lowdegree}

The log-density approach compares frequencies of small patterns in random and planted models~\cite{BCCFV10}. Related work includes hierarchy gaps for densest subgraph~\cite{bhaskara2012polynomial}, approximation bounds for label cover~\cite{chlamtavc2017approximation}, and hypergraphic planted clique~\cite{luo2020open}. We adapt the low-degree moment calculation of Dhawan, Mao, and Wein for a random dense subhypergraph~\cite{dhawan2025detection} to an arbitrary fixed expanding plant. Expansion gives a lower bound on the number of vertices that must match for two random copies to share a given set of edges. In both planting models, the mean of any real polynomial of degree at most $\lfloor n^{1-\varepsilon}\rfloor$ in the edge indicators shifts by $o(1)$ null standard deviations.

This degree grows as a positive power of the number $N$ of public vertices. Under the low-degree heuristic, this is evidence for subexponential hardness~\cite{kunisky2019notes}. The heuristic draws on connections with spectral algorithms, sum-of-squares, and statistical queries~\cite{hopkins2017power,brennan2020statistical}. Rigorous implications under noise are known for fully permutation-invariant vectors~\cite{hsieh2026rigorous} and for planting in sparse random graphs~\cite{chen2026rigorous}. Among the counterexamples to broader formulations~\cite{holmgren2020counterexamples,buhai2025quasi} is Mao's graph model, where Reed--Muller structure yields a finite-field rank test~\cite{mao2026polynomial}.

For ordinary graphs, Yu, Zadik, and Zhang characterize constant-degree detection of a fixed plant when the background density is bounded below~\cite{yu2024counting}. Elimelech and Huleihel treat both dense and sparse backgrounds~\cite{elimelech2025detecting}. In a related cryptographic model, a random graph and a larger graph containing a hidden induced copy are both public. Abram et al.\ use this model for private simultaneous messages and secret sharing~\cite{abram2023planted}. Bogdanov et al.\ give low-degree bounds over a wider range of parameters, including for hypergraphs~\cite{bogdanov2024low}. In our model planting prescribes only the edges of the fixed hypergraph.

\paragraph{Measuring Low-Degree Advantage.}
Fix a simple $k$-uniform hypergraph $H$ with $n$ vertices and $m$ edges. Let $\mathcal Q$ be the law of $\mathcal G_k(N,p)$ and let $\mathcal P_H$ be the law obtained by adding a uniformly injected copy of $H$. The input is its edge-indicator vector $\mathbf{Y}\in\{0,1\}^{\binom Nk}$. The coefficients of a polynomial test may depend arbitrarily on $H$.

For a degree cutoff $D$, we measure a real polynomial $f$ by its normalized difference of means,
\[
 \frac{|\mathbb E_{\mathcal P_H}f-\mathbb E_{\mathcal Q}f|}
 {\sqrt{\operatorname{Var}_{\mathcal Q}(f)}}.
\]
The maximum over such $f$ of degree at most $D$ with positive null variance is the \emph{degree-$D$ advantage}~\cite{kunisky2019notes}. Thresholding $f$ can increase its degree as a function of the edge indicators.

\begin{theorem}\label{thm:low-degree}
Let $D\ge1$ be an integer, $0<p<1$, and $N\ge n$. Suppose every set $F\subseteq E(H)$ with $|F|\le D$ spans at least $\alpha k|F|$ vertices, where $0<\alpha\le1$, and suppose $kD<N-1$. The degree-$D$ advantage is at most
\begin{equation}\label{eq:ld-bound}
 \left[\exp\left(\frac{k!m^2(1-p)}{p(N-kD)^{\alpha k}}\right)-1\right]^{1/2}.
\end{equation}
\end{theorem}
\begin{proof}
Under $\mathcal Q$, the products $Z_S=\prod_{e\in S}Z_e$, with $Z_e=(Y_e-p)/\sqrt{p(1-p)}$, form the orthonormal basis for the product Bernoulli law. A degree-$D$ polynomial uses only $|S|\le D$, by multilinear reduction on Boolean inputs. After centering under $\mathcal Q$, its variance is the squared norm of its coefficient vector, and its planted mean is the inner product with the vector of planted moments. Optimizing over the coefficient vector gives the square of the degree-$D$ advantage,
\begin{equation}\label{eq:ld-moments}
 \sum_{1\le |S|\le D}(\mathbb E_{\mathcal P_H}Z_S)^2.
\end{equation}

Let $\pi:[n]\hookrightarrow[N]$ be the planted injection. Conditional on $\pi$, a product $Z_S$ has mean zero unless every edge in $S$ is forced by $\pi(H)$: otherwise one of its factors is an independent centered variable. Therefore
\[
 \mathbb E_{\mathcal P_H}Z_S
 =\left(\frac{1-p}{p}\right)^{|S|/2}
   \Pr[S\subseteq E(\pi(H))].
\]
Squaring the probability amounts to requiring the same set $S$ to lie in two independent planted copies. For two independent uniform injections $\pi,\pi'$, the contribution of sets of exactly $\ell$ edges to~\eqref{eq:ld-moments} is consequently
\begin{equation}\label{eq:ld-overlap}
 \left(\frac{1-p}{p}\right)^\ell
 \mathbb E_{\pi,\pi'}\binom{|E(\pi(H))\cap E(\pi'(H))|}{\ell}.
\end{equation}

Each common $\ell$-edge set has unique preimages under $\pi$ and $\pi'$. Thus the expectation in~\eqref{eq:ld-overlap} is the sum of $\Pr[\pi(F)=\pi'(F')]$ over all pairs of $\ell$-edge subsets $F,F'\subseteq E(H)$. Only isomorphic pairs contribute. Write $v$ for the number of vertices spanned by each. To specify an isomorphism, match the $\ell$ edges and then match the $k$ vertices within each pair. This gives at most $\ell!(k!)^\ell$ choices since the matches must agree on shared vertices. Conditional on $\pi$, each isomorphism fixes the images of $v$ vertices under $\pi'$, an event of probability $1/(N)_v$, where $(N)_v=N(N-1)\cdots(N-v+1)$.

Expansion gives $v\ge\alpha k\ell$, while
$v\le k\ell\le kD$. Since $N-kD>1$, every factor in the falling factorial
is at least $N-kD$, and hence
$(N)_v\ge(N-kD)^{\alpha k\ell}$. There are
$\binom m\ell^2$ choices of $(F,F')$, so~\eqref{eq:ld-overlap} is at most
\[
 \left(\frac{1-p}{p}\right)^\ell
 \binom m\ell^2\frac{\ell!(k!)^\ell}{(N-kD)^{\alpha k\ell}}
 \le\frac1{\ell!}
 \left(\frac{k!m^2(1-p)}{p(N-kD)^{\alpha k}}\right)^\ell.
\]
Summing over $1\le\ell\le D$ and bounding by the exponential series
proves~\eqref{eq:ld-bound}.
\end{proof}

\begin{corollary}\label{cor:ld-parameters}
Fix a positive integer $k$ and constants $s,\beta>0$, $0<\varepsilon<1/2$, and $C>1$, with
\[
 \beta<(1-\varepsilon)k-\frac{2s}{C}.
\]
For any deterministic sequence $m\le n^{s+o(1)}$, if $N=\lceil n^C\rceil$, $p=N^{-\beta+o(1)}$, and $H$ is $(1-\varepsilon,k,n^{1-\varepsilon})$-expanding, then its degree-$\lfloor n^{1-\varepsilon}\rfloor$ advantage tends to zero, uniformly over $H$.
\end{corollary}
\begin{proof}
Take $D=\lfloor n^{1-\varepsilon}\rfloor$. Then $kD=o(N)$ and the exponent inside~\eqref{eq:ld-bound} is at most
\[
 n^{2s+C\beta-C(1-\varepsilon)k+o(1)}=o(1).
\]
\end{proof}

\begin{remark}[Marked edges]\label{rem:marked-low-degree}
Theorem~\ref{thm:low-degree} and Corollary~\ref{cor:ld-parameters} hold with the same bounds for marked edges. The null indicators of $(e,L)$ are independent, and in the overlap count requiring each vertex bijection to preserve the marks only reduces the number of choices. The number of spanned vertices and the injection probability $1/(N)_v$ are unchanged.

\end{remark}

\paragraph{The Parameters for Planting.}
In Conjecture~\ref{conj:plant}, $k\lfloor n^{1-\varepsilon}\rfloor=o(N)$. The exponent in~\eqref{eq:ld-bound} is therefore $O_k(m^2/(pN^{(1-\varepsilon)k}))$. With the stated edge budget this is $O(n^{[2-C_1(C_2-1)]\varepsilon k})$, which tends to zero since $C_1>2$ and $C_2\ge2$. In Conjecture~\ref{conj:marked-plant} the ratio is at most $2n^{-8\varepsilon k}$.

The expansion bound~\eqref{eq:expand-edge-budget} sharpens this conclusion. Using it in the unmarked ratio gives $O_{k,\varepsilon}(n^{2+2\varepsilon^2k-C_1(C_2-1)\varepsilon k})$. Thus every $C_2>1$ suffices for all sufficiently small fixed $\varepsilon$ and then all sufficiently large fixed $k$. Conversely, at any fixed $C_2\le1$ the proposed distinguishing bound fails for an admissible plant. Choose a permitted rational $\varepsilon$ and arbitrarily large $k$ with $j=2\varepsilon k$ integral. Take $H$ to consist of two edges sharing $j$ vertices and $n-2k+j$ isolated vertices. It meets the expansion requirement. If $C_2\le1$, the null expectation of the number of unordered pairs sharing at least $j$ vertices is
\[
 \sum_{i=j}^{k-1}\frac{(N)_{2k-i}p^2}{2\,i!\,((k-i)!)^2}
 \ \le\ \frac{1+o(1)}{2\,j!\,((k-j)!)^2}<\frac14 .
\]
The planted graph always contains such a pair, and for fixed $k$ counting pairs takes polynomial time.

More generally, suppose $m\le n^{s+o(1)}$, $N=\lceil n^C\rceil$ and $p=N^{-\beta+o(1)}$ as in Corollary~\ref{cor:ld-parameters}. For a fixed noisy-$k$XOR witness $\delta$, the strict range
\[
 (1-\delta)k<\beta<(1-\varepsilon)k-\frac{2s}{C}
\]
both keeps the public edge count within the $N^{\delta k}$-row budget with high probability and gives the low-degree lower bound. For the density in~\eqref{eq:plant-parameters}, the choice of $\varepsilon$ in Section~\ref{sec:pke} gives the strict lower inequality.

In either edge model the edge count attains the exact squared degree-one advantage $(1-p)m^2/(p|\Omega|)$. Indeed every edge has the same planting probability $m/|\Omega|$, so the degree-one moments in~\eqref{eq:ld-moments} are equal.

\clearpage
\bibliographystyle{alpha}
\bibliography{refs}
\end{document}